\documentclass[12pt]{article}
\usepackage[english]{babel}

\usepackage{fullpage}
\usepackage{cite}

\usepackage{newpxtext,newpxmath}

\let\coloneqq\relax

\usepackage[utf8]{inputenc}
\usepackage{amsthm}
\usepackage{amssymb}
\usepackage{amsmath}
\usepackage{bbold}
\usepackage{bbm}
\usepackage[pdftex, backref=page]{hyperref}
\usepackage{braket}
\usepackage{dsfont}
\usepackage{mathdots}
\usepackage{mathtools}
\usepackage{enumerate}
\usepackage[shortlabels]{enumitem}
\usepackage{csquotes}
\usepackage{stmaryrd}
\usepackage[cal=boondox]{mathalfa}
\usepackage{graphicx}
\usepackage{stackengine}
\usepackage{scalerel}
\usepackage{tensor}       
\usepackage{array}
\usepackage{makecell}
\newcolumntype{x}[1]{>{\centering\arraybackslash}p{#1}}
\usepackage{tikz}
\usepackage{pgfplots}
\usetikzlibrary{shapes.geometric, shapes.misc, positioning, arrows, arrows.meta, decorations.pathreplacing, decorations.pathmorphing, patterns, angles, quotes, calc}
\usepackage{booktabs}
\usepackage{xfrac}
\usepackage{siunitx}
\usepackage{centernot}
\usepackage{comment}
\usepackage{chngcntr}
\usepackage{caption}
\usepackage{subcaption}

\newtheorem{thm}{Theorem}
\newtheorem*{thm*}{Theorem}
\newtheorem{prop}[thm]{Proposition}
\newtheorem*{prop*}{Proposition}
\newtheorem{lemma}[thm]{Lemma}
\newtheorem*{lemma*}{Lemma}
\newtheorem{cor}[thm]{Corollary}
\newtheorem*{cor*}{Corollary}

\newtheorem*{cj*}{Conjecture}
\newtheorem{Def}[thm]{Definition}
\newtheorem*{Def*}{Definition}

\newtheorem*{question*}{Question}

\newtheorem*{problem*}{Problem}

\makeatletter
\def\thmhead@plain#1#2#3{%
  \thmname{#1}\thmnumber{\@ifnotempty{#1}{ }\@upn{#2}}%
  \thmnote{ {\the\thm@notefont#3}}}
\let\thmhead\thmhead@plain
\makeatother

\theoremstyle{definition}
\newtheorem{rem}[thm]{Remark}

\newtheorem{defi}[thm]{Definition}

\newcommand{\bb}{\begin{equation}\begin{aligned}\hspace{0pt}}
\newcommand{\bbb}{\begin{equation*}\begin{aligned}}
\newcommand{\ee}{\end{aligned}\end{equation}}
\newcommand{\eee}{\end{aligned}\end{equation*}}
\newcommand*{\coloneqq}{\mathrel{\vcenter{\baselineskip0.5ex \lineskiplimit0pt \hbox{\scriptsize.}\hbox{\scriptsize.}}} =}

\newcommand{\eqt}[1]{\stackrel{\mathclap{\scriptsize \mbox{#1}}}{=}}

\newcommand{\geqt}[1]{\stackrel{\mathclap{\scriptsize \mbox{#1}}}{\geq}}
\newcommand{\ketbra}[1]{\ket{#1}\!\!\bra{#1}}

\newcommand{\ketbraa}[2]{\ket{#1}\!\!\bra{#2}}

\newcommand{\sumno}{\sum\nolimits}

\renewcommand{\epsilon}{\varepsilon}

\newcommand{\id}{\mathds{1}}
\newcommand{\R}{\mathds{R}}

\DeclareMathOperator{\Tr}{Tr}

\DeclareMathAlphabet{\pazocal}{OMS}{zplm}{m}{n}

\DeclareMathOperator{\supp}{supp}

\DeclareMathOperator{\Id}{id}

\newcommand{\HH}{\mathcal{H}}

\newcommand{\NN}{\pazocal{N}}
\newcommand{\EE}{\pazocal{E}}

\newcommand{\lsmatrix}{\left(\begin{smallmatrix}}
\newcommand{\rsmatrix}{\end{smallmatrix}\right)}

\stackMath

\newcommand{\rel}[3]{#1\big(#2\,\big\|\,#3\big)}
\newcommand{\Rel}[3]{#1\Big(#2\,\Big\|\,#3\Big)}

\stackMath

\makeatletter
\newcommand*\rel@kern[1]{\kern#1\dimexpr\macc@kerna}
\newcommand*\widebar[1]{%
  \begingroup
  \def\mathaccent##1##2{%
    \rel@kern{0.8}%
    \overline{\rel@kern{-0.8}\macc@nucleus\rel@kern{0.2}}%
    \rel@kern{-0.2}%
  }%
  \macc@depth\@ne
  \let\math@bgroup\@empty \let\math@egroup\macc@set@skewchar
  \mathsurround\z@ \frozen@everymath{\mathgroup\macc@group\relax}%
  \macc@set@skewchar\relax
  \let\mathaccentV\macc@nested@a
  \macc@nested@a\relax111{#1}%
  \endgroup
}

\counterwithin*{equation}{part}
\counterwithin*{thm}{part}
\counterwithin*{figure}{part}

\tikzset{meter/.append style={draw, inner sep=10, rectangle, font=\vphantom{A}, minimum width=30, line width=.8, path picture={\draw[black] ([shift={(.1,.3)}]path picture bounding box.south west) to[bend left=50] ([shift={(-.1,.3)}]path picture bounding box.south east);\draw[black,-latex] ([shift={(0,.1)}]path picture bounding box.south) -- ([shift={(.3,-.1)}]path picture bounding box.north);}}}
\tikzset{roundnode/.append style={circle, draw=black, fill=gray!20, thick, minimum size=10mm}}
\tikzset{squarenode/.style={rectangle, draw=black, fill=none, thick, minimum size=10mm}}

\definecolor{Blues5seq1}{RGB}{239,243,255}
\definecolor{Blues5seq2}{RGB}{189,215,231}
\definecolor{Blues5seq3}{RGB}{107,174,214}
\definecolor{Blues5seq4}{RGB}{49,130,189}
\definecolor{Blues5seq5}{RGB}{8,81,156}

\definecolor{Greens5seq1}{RGB}{237,248,233}
\definecolor{Greens5seq2}{RGB}{186,228,179}
\definecolor{Greens5seq3}{RGB}{116,196,118}
\definecolor{Greens5seq4}{RGB}{49,163,84}
\definecolor{Greens5seq5}{RGB}{0,109,44}

\definecolor{Reds5seq1}{RGB}{254,229,217}
\definecolor{Reds5seq2}{RGB}{252,174,145}
\definecolor{Reds5seq3}{RGB}{251,106,74}
\definecolor{Reds5seq4}{RGB}{222,45,38}
\definecolor{Reds5seq5}{RGB}{165,15,21}

\allowdisplaybreaks

\usepackage[most,breakable]{tcolorbox}
\newenvironment{boxedthm}[1]%
	{\expandafter\ifstrequal\expandafter{#1}{orange}{\begin{tcolorbox}[colback=red!15,colframe=orange!15,breakable,enhanced]}{\begin{tcolorbox}[colback=Blues5seq1,colframe=Blues5seq5,breakable,enhanced]}}%
	{\end{tcolorbox}}

\usepackage{authblk}

\renewcommand{\EE}[1]{\underset{\scaleobj{.8}{#1}}{\mathds{E}\,}}

\allowdisplaybreaks[1] 
\usepackage{graphicx,import,xcolor,transparent} 
\usepackage{enumitem}

\begin{document}

\author[1]{Filippo Girardi$^\clubsuit$\thanks{\texttt{filippo.girardi@sns.it}\hfill $^\clubsuit$These authors contributed equally.}}
\author[1]{Francesco Anna Mele$^\clubsuit$\thanks{\texttt{francesco.mele@sns.it}\hfill $^\diamondsuit$These authors contributed equally.}}
\author[2]{Haimeng Zhao$^\clubsuit$\thanks{\texttt{haimengzhao@icloud.com}}}
\author[3]{Marco~Fanizza$^\diamondsuit$\thanks{\texttt{marco.fanizza@inria.fr}}}
\author[1]{Ludovico~Lami$^\diamondsuit$\thanks{\texttt{ludovico.lami@gmail.com}}}
\affil[1]{Scuola Normale Superiore, Piazza dei Cavalieri 7, 56126 Pisa, Italy}
\affil[2]{\mbox{Institute for Quantum Information and Matter, California Institute of Technology, Pasadena, CA 91125, USA}}
\affil[3]{Inria, T\'el\'ecom Paris -- LTCI, Institut Polytechnique de Paris, Palaiseau, France}

\title{\vspace*{-0.5em}Random Stinespring superchannel: \\ converting channel queries into \\ dilation isometry queries}

\date{}
\setcounter{Maxaffil}{0}
\renewcommand\Affilfont{\itshape\small}

\maketitle\vspace{-5ex}
\begin{abstract} 
The recently introduced random purification channel, which converts $n$ copies of an arbitrary mixed quantum state into $n$ copies of the same uniformly random purification, has emerged as a powerful tool in quantum information theory. Motivated by this development, we introduce a channel-level analogue, which we call the \emph{random Stinespring superchannel}. This consists in a procedure to transform $n$ parallel queries of an arbitrary quantum channel into $n$ parallel queries of the same uniformly random Stinespring isometry, via universal encoding and decoding operations that are efficiently implementable.  When the channel is promised to have Choi rank at most $r$, the procedure can be tailored to yield a Stinespring environment of dimension $r$. We present two applications of the random Stinespring superchannel, one in quantum Shannon theory and one in quantum learning theory. In quantum Shannon theory, we prove a channel-level analogue of Uhlmann’s theorem for quantum divergences. In quantum learning theory, our construction shows that tomography of quantum channels reduces to tomography of isometries. This yields a simple channel learning algorithm, based on existing isometry learning protocols, that matches the performance of the two recently proposed channel tomography algorithms.
Complementarily, whereas the optimality of these algorithms had previously been established only up to a logarithmic factor in
the dimension, we close this gap by removing this logarithmic factor from the
lower bound. Taken together, our results fully establish the optimality of
these recently introduced channel learning algorithms, showing that the
optimal query complexity of learning a quantum channel with input dimension
$d_A$, output dimension $d_B$, and Choi rank $r$ is $\Theta(d_A d_B r)$.\\
\end{abstract}


\section{Introduction}
The recently introduced random purification channel~\cite{tang2025, random_pur_simple}, which converts $n$ copies of a mixed quantum state $\rho$ into $n$ copies of the same randomly chosen purification of $\rho$, has already proved to be a very powerful tool in quantum information theory, with applications spanning quantum learning theory~\cite{pelecanos2025,Utsumi2025, AMele2025, WalterWitteveen_2025}, quantum Shannon theory~\cite{random_pur_simple}, and Gaussian quantum information~\cite{WalterWitteveen_2025, cv_purification}. Notably, the random purification channel admits a remarkably simple analytic form~\cite{random_pur_simple} and can also be implemented efficiently using quantum circuits~\cite{tang2025, pelecanos2025}. More precisely, for any Hilbert space $\mathcal{H}_A$ and any integer $n \geq 1$, there exists a quantum channel
\bb
\Lambda_{\rm purify}^{(n)}:\mathcal{L}(\mathcal{H}_A^{\otimes n})
\to \mathcal{L}\big((\mathcal{H}_A \otimes \mathcal{H}_B)^{\otimes n}\big),
\ee
where $\mathcal{H}_B$ is isomorphic to $\mathcal{H}_A$ and $\mathcal{L}(\HH)$ denotes the space of linear operators on $\mathcal{H}$, such that, for all states
$\rho_A \in \mathcal{D}(\mathcal{H}_A)$, one has~\cite{tang2025, random_pur_simple}
\bb \label{eq:property}
\Lambda^{(n)}_{\mathrm{purify}}(\rho_A^{\otimes n})
=
\EE{U_B}\!\left[
(\id_A \otimes U_B)\,
(\psi_\rho)_{AB}\,
(\id_A \otimes U_B^\dagger)
\right]^{\otimes n}.
\ee
where the expectation value is taken over Haar-random unitaries $U_B$ acting on
$\mathcal{H}_B$, $(\psi_\rho)_{AB}$ denotes an arbitrary fixed purification of
$\rho_A$ in $\mathcal{H}_A \otimes \mathcal{H}_B$, and $\id_A$ is the identity operator on $\mathcal{H}_A$. In other words, this channel transforms $n$ copies of $\rho_A$ into $n$ copies of a uniformly random purification of $\rho_A$. A simpler formula describing its action is~\cite{random_pur_simple}
\begin{equation}\label{compact_form}
\Lambda^{(n)}_{\mathrm{purify}}(\,\cdot\,)
=
\sqrt{R_n}\,
\bigl(\;\cdot\, \otimes \id_{B}^{\otimes n}\bigr)\,
\sqrt{R_n}\,,
\end{equation}
where $R_n \coloneqq \EE{U_B}\!\left[ \bigl(\id_{A} \otimes U_B \bigr)\, \Gamma_{AB}\, \bigl(\id_{A}  \otimes U_B^\dagger\bigr) \right]^{\otimes n}$, and $\Gamma_{AB} \coloneqq \sum_{i,j} \ketbraa{i}{j}_A\otimes \ketbraa{i}{j}_B$ is the unnormalised maximally entangled state.

The concept of purification of a state is only the first instance of the general idea that quantum information manipulation can be conceptually simplified by enlarging the underlying Hilbert space, an attitude colloquially known as the Church of the Larger Hilbert Space. Following this train of thought, the next logical step is the purification of quantum channels, called Stinespring dilation~\cite{Stinespring}:
for any quantum channel
$\Phi_{A\to B}:\mathcal{L}(\mathcal{H}_A)\to \mathcal{L}(\mathcal{H}_B)$, there
exists a Hilbert space $\mathcal{H}_E$, representing the environment, and an isometry
$V_{A\to BE}:\mathcal{H}_A \to \mathcal{H}_B \otimes \mathcal{H}_E$, called the Stinespring
isometry, such that
\bb
\Phi_{A\to B}(\,\cdot\,)=\Tr_E\!\left[V_{A\to BE}(\,\cdot\,)V_{A\to BE}^\dagger\right].
\ee
In particular, letting $d_A \coloneqq \dim \mathcal{H}_A$ and $d_B \coloneqq \dim \mathcal{H}_B$ denote the input and output dimensions of the
channel, the environment $\mathcal{H}_E$ can always be chosen to have dimension
$d_E = d_A d_B$. More generally, if the channel is promised to have \emph{Choi rank} $r$, defined
as the rank of the associated Choi state, then the environment can be taken to
have dimension $d_E = r \le d_A d_B$.

In the same spirit as for the random purification channel, one may therefore ask the
following question:
\begin{center}
    \emph{
    Can $n$ queries of a quantum channel $\Phi_{A\to B}$ be converted into $n$
    queries of a randomly chosen Stinespring isometry $V_{A\to BE}$?
    } 
\end{center}
In this paper, we answer this
question in the affirmative by exhibiting a procedure, called the
\emph{random Stinespring superchannel}---that converts $n$ parallel uses of a quantum channel
(that is, a single query of $\Phi_{A\to B}^{\otimes n}$) into $n$ parallel uses of
the same random Stinespring isometry (that is, a single query of
$V_{A\to BE}^{\otimes n}$). Notably, we also show that this procedure can be
implemented efficiently in terms of a quantum circuit. 
Moreover, our result proves Conjecture 1.8 of~\cite{tang2025} in the parallel-query setting and confirms the intuition, suggested by the analysis of~\cite{chen2025quantumchanneltomographyestimation}, that a random Stinespring superchannel should exist. More precisely, our main result is the following.
\begin{boxedthm}{}
\begin{thm}[(Random Stinespring superchannel)]\label{thm1}
Let $\mathcal{H}_A$ and $\mathcal{H}_B$ be Hilbert spaces of dimensions $d_A$ and $d_B$, respectively, and let $n,r\ge1$ such that $r\le d_Ad_B$. There exist an environment Hilbert space $\mathcal{H}_E$ of dimension $r$, an auxiliary Hilbert space $\mathcal{H}_M$, an encoding quantum channel
\[
\mathcal{E}:\mathcal{L}(\mathcal{H}_A^{\otimes n})
\to \mathcal{L}(\mathcal{H}_A^{\otimes n}\otimes \mathcal{H}_M),
\]
and a decoding quantum channel
\[
\mathcal{D}:\mathcal{L}(\mathcal{H}_B^{\otimes n}\otimes \mathcal{H}_M)
\to \mathcal{L}\left(\mathcal{H}_{B}^{\otimes n}\otimes\mathcal{H}_E^{\otimes n}\right),
\]
such that for every quantum channel
$\Phi:\mathcal{L}(\mathcal{H}_A)\to \mathcal{L}(\mathcal{H}_B)$ with Choi rank at
most $r$ it holds that
\bb\label{eq_random_isometry}
\big(\mathcal{D}\circ\left(\Phi^{\otimes n}\otimes {\rm Id}_M\right)\circ \mathcal{E}\big)(\,\cdot\,)
=
\EE{U_E}\!\left[
\big((\id_B\otimes U_E)V_{A\to BE}\big)^{\otimes n}
(\,\cdot\,)
\big(V_{A\to BE}^\dagger(\id_B\otimes U_E^\dagger)\big)^{\otimes n}
\right],
\ee
where the expectation value is taken over Haar-random unitaries $U_E$ acting on
$\mathcal{H}_E$, and $V_{A\to BE}$ is any fixed Stinespring isometry associated with
$\Phi$. In addition, both $\mathcal{E}$ and $\mathcal{D}$ can be implemented in
polynomial time in $n$ and in the logarithm of the dimensions of the Hilbert spaces
involved. In other words, $n$ parallel queries of $\Phi$ can be efficiently
converted into $n$ parallel queries of a uniformly random Stinespring isometry.
\end{thm}
\end{boxedthm}
\begin{figure}[t]
  \centering
  \def\svgwidth{\linewidth}
  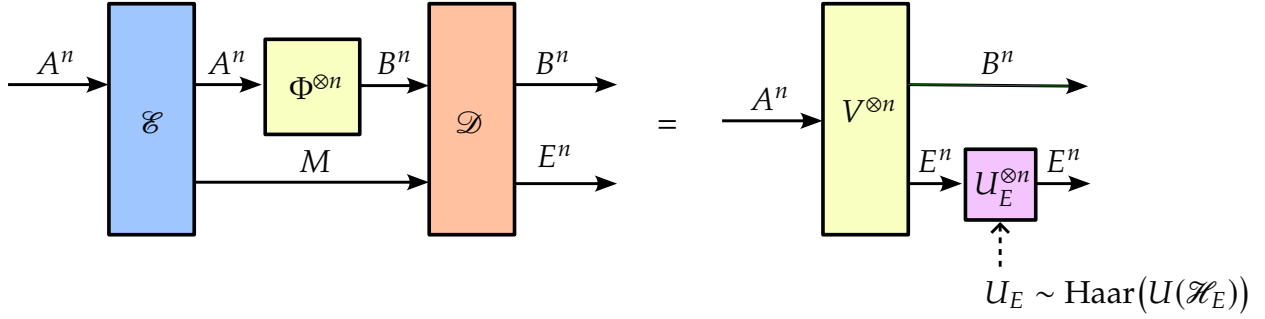
  \caption{Schematic representation of the random Stinespring superchannel introduced in Theorem~\ref{thm1}.}
  \label{fig:figure}
\end{figure}
The proof of this theorem is provided at the end of Section~\ref{explicit_circuit}. The random Stinespring superchannel is illustrated schematically in
Fig.~\ref{fig:figure}. At first glance, one might be tempted to think that the
procedure could be implemented by choosing the encoding channel
$\mathcal{E}={\rm Id}_{A\to A}$ and the decoding channel
$\mathcal{D}=\Lambda_{\rm purify}^{(n)}$, namely the random purification channel
defined in~\eqref{eq:property}. However, this naive approach fails, as $\Lambda_{\rm purify}^{(n)}$ automatically symmetrises its input, meaning that Eq.~\eqref{eq_random_isometry} would not be satisfied. (Another, more intuitive way to think about this is that the random purification channel at the output would also purify the input mixed states, which is not what Eq.~\eqref{eq_random_isometry} does.) In fact, our construction does not employ the
random purification channel as a subroutine. 

Crucially, the encoding and decoding channels we construct are independent of the input state and of whatever post-processing to which the output may be subjected. 
Our random Stinespring superchannel is therefore \emph{universal} and plays, at the level of quantum channels, the same conceptual role the random purification channel plays for quantum states. When the channel is a replacement channel that prepares a mixed state, our random Stinespring superchannel recovers the random purification channel, albeit with a larger environment. 

We present two applications of the random Stinespring superchannel: one in quantum Shannon theory and one in quantum learning theory.  On the quantum Shannon theory side, we use the random Stinespring superchannel to extend Uhlmann's theorem for quantum divergences~\cite{Mazzola_2025, Fang2025-variational,random_pur_simple} from quantum states to quantum channels. On the quantum learning theory side, our result is closely connected to the problem of quantum channel learning~\cite{AMele2025, chen2025quantumchanneltomographyestimation}.
Indeed, in the same spirit as Ref.~\cite{pelecanos2025}, where the random purification channel was used to show that mixed-state learning reduces to pure-state learning 
(and later generalised in Ref.~\cite{AMele2025} to show that quantum channel learning reduces to learning a purification of the Choi state), our construction immediately implies that \emph{quantum channel learning reduces to isometry learning}, specifically to learning a Stinespring isometry of the channel. This observation has been used very recently in~\cite{chen2025quantumchanneltomographyestimation} to provide another proof of the previously established formula for the query complexity of quantum channel learning, up to logarithmic dimensional factors~\cite{AMele2025}. Specifically, our procedure --- as well as 
that of~\cite{chen2025quantumchanneltomographyestimation} --- implies that tomography of quantum channels with input dimension $d_A$, output dimension $d_B$, and Choi rank $r$ reduces to tomography of isometries with input dimension $d_A$ and output dimension $d_B r$. By leveraging the upper bound on the query complexity of isometry learning found in~\cite{AMele2025,
chen2025quantumchanneltomographyestimation}, one readily obtains the upper bound $O(d_A d_B r)$ on the query complexity of quantum channel learning, which was recently established by~\cite{AMele2025} for the first time and was known to be optimal up to logarithmic dimensional factors. As a complementary result, we also prove a matching lower bound $\Omega(d_A d_B r)$, which holds even against the most general classes of queries, including those with inverse and controlled queries and indefinite causal order. 
We prove this lower bound by developing a proof technique that is purely algebraic and reinforces the simple intuition from dimension counting, completely circumventing the heavy representation theory machinery previously used for unitaries~\cite{haah2023query,bavaresco2022unitary}.
Taken together, these results establish that the optimal query complexity for quantum channel tomography is $\Theta(d_A d_B r)$, thereby removing the remaining logarithmic gap in the lower bound.
\bigskip

The paper is organised as follows. In Section~\ref{sec_existence}, we present a remarkably simple proof of Theorem~\ref{thm1}. In Section~\ref{explicit_circuit}, we describe an explicit and efficient quantum circuit implementation of the random Stinespring superchannel. In Section~\ref{appl_Shannon}, we apply this superchannel to quantum Shannon theory and derive an extension of Uhlmann's theorem for quantum divergences~\cite{Mazzola_2025,Fang2025-variational,random_pur_simple} to quantum channels. In Section~\ref{appl_learning}, we apply it to quantum learning theory, focusing on the problem of tomography of quantum channels ~\cite{AMele2025, chen2025quantumchanneltomographyestimation}, and derive the optimal query complexity of learning quantum channels without additional logarithmic factors. 
Finally, in Section~\ref{sec_con}, we summarise our results and outline several open problems
for future work.

\section{{A simple proof of Theorem~\ref{thm1} for $\boldsymbol{r=d_Ad_B}$}}\label{sec_existence}

This section presents a simple proof of the first part of Theorem~\ref{thm1} in the special
case of $r = d_A d_B$, without dealing with the implementation efficiency. More precisely, the goal of this section is to prove the existence of a
\emph{physical} supermap implementing the random Stinespring superchannel. Concretely, this amounts to showing that there exist an
encoding channel $\mathcal{E}$, a memory system $M$, and a decoding channel
$\mathcal{D}$ such that, for any quantum channel $\Phi_{A\to B}$, Eq.~\eqref{eq_random_isometry} holds,
where $V_{A\to BE}$ denotes a fixed Stinespring isometry of $\Phi$.

To address this question, we invoke the formalism of \emph{superchannels} introduced in Ref.~\cite{Chiribella2008}. By definition, a superchannel is a supermap $\mathcal{S}$ that maps quantum channels into quantum channels in a completely positive way. Formally, a linear map $\mathcal{S}$ taking as input maps $\mathcal{L}(\mathcal{H}_A)\to \mathcal{L}(\mathcal{H}_B)$ and outputting maps $\mathcal{L}(\mathcal{H}_{\tilde{A}})\to \mathcal{L}(\mathcal{H}_{\tilde{B}})$ is said to be a superchannel if: (a)~it is completely positive, in the sense that, for all completely positive maps $\NN:\mathcal{L}(\mathcal{H}_{A}\otimes \mathcal{H}_{E})\to \mathcal{L}(\mathcal{H}_{B} \otimes \mathcal{H}_F)$, where $E$ and $F$ are auxiliary quantum systems, the transformed map\footnote{To define the action of $\mathcal{S}\otimes \Id$, note that $\mathcal{L}\left(\mathcal{L}(\mathcal{H}_{A}\otimes \mathcal{H}_{E})\to \mathcal{L}(\mathcal{H}_{B} \otimes \mathcal{H}_F)\right)$ is canonically isomorphic to $\mathcal{L}\big(\mathcal{L}(\mathcal{H}_A)\to \mathcal{L}(\mathcal{H}_B)\big) \otimes \mathcal{L}\big(\mathcal{L}(\mathcal{H}_E)\to \mathcal{L}(\mathcal{H}_F)\big)$. We think of $\mathcal{S}$ as acting on the first tensor factor, and of $\Id$ as the identity operator acting on the second.} $(\mathcal{S}\otimes \Id)[\NN]: \mathcal{L}(\mathcal{H}_{\tilde{A}}\otimes \mathcal{H}_{E})\to \mathcal{L}(\mathcal{H}_{\tilde{B}} \otimes \mathcal{H}_F)$ is again completely positive; and (b)~it maps trace-preserving maps to trace-preserving maps. 

To simplify the picture, we can look at the action of $\mathcal{S}$ at the level of Choi operators. Denoting the associated map with $\mathcal{S}_\ast: \mathcal{L}(\mathcal{H}_A\otimes \mathcal{H}_B) \to \mathcal{L}\big(\mathcal{H}_{\tilde{A}}\otimes \mathcal{H}_{\tilde{B}}\big)$, condition~(a) is equivalent to requiring that $\mathcal{S}_\ast$ is completely positive, and condition~(b) is equivalent to demanding that it sends operators $X_{AB}$ on $\mathcal{H}_A\otimes \mathcal{H}_B$ such that $\Tr_B X_{AB} = \kappa \id_A$, for some $\kappa\in \R$, to operators $Y_{\tilde{A}\tilde{B}}$ such that $\Tr_{\tilde{B}} Y_{\tilde{A}\tilde{B}} = \kappa \id_{\tilde{A}}$. 

A central result of Ref.~\cite{Chiribella2008} provides a useful characterisation of superchannels.

\begin{lemma}[{\cite[Theorem 1]{Chiribella2008}}] \label{lemma_giulio}
For any superchannel $\mathcal{S}$ there exist an encoding channel $\mathcal{E}$, a memory system $M$, and a decoding channel $\mathcal{D}$ such that, for any quantum channel $\Phi$,
\bb
\mathcal{S}[\Phi]
=
\mathcal{D}\circ\left(\Phi\otimes {\rm Id}_M\right)\circ \mathcal{E}.
\ee
\end{lemma}
In light of Lemma~\ref{lemma_giulio}, proving Theorem~\ref{thm1} reduces to
establishing that the mapping
\bb\label{eq:mapping}
\Phi^{\otimes n}(\,\cdot\,)\quad \longmapsto\quad 
\EE{U_E}\!\left[
\big((\id_B\otimes U_E)V_{A\to BE}\big)^{\otimes n}
(\,\cdot\,)
\big(V_{A\to BE}^\dagger(\id_B\otimes U_E^\dagger)\big)^{\otimes n}
\right]
\ee
is implementable as a superchannel. 
The following proposition, combined with Lemma~\ref{lemma_giulio}, ensures that, for any $n\geq 1$, there exists a superchannel $\mathcal{S}^{(n)}$ that implements~\eqref{eq:mapping}.

\begin{prop}\label{prop:Choi}
Let us consider the linear map $\mathcal{S}^{(n)}_\ast: \mathcal{L}(\mathcal{H}_{A^n}\otimes \mathcal{H}_{B^n}) \to \mathcal{L}\big(\mathcal{H}_{A^n}\otimes \mathcal{H}_{(BE)^n}\big)$, where $n\geq 1$ is an arbitrary integer, defined as
\bb\label{eq:S_ast}
    \mathcal{S}^{(n)}_\ast(X_{(AB)^n})\coloneqq \Lambda^{(n)}_{\rm purify}\left(X_{(AB)^n}\right).
\ee
Then there exists a superchannel $\mathcal{S}^{(n)}$ taking as input maps $\mathcal{L}(\mathcal{H}_{A^n})\to \mathcal{L}(\mathcal{H}_{B^n})$ and outputting maps $\mathcal{L}(\mathcal{H}_{A^n})\to \mathcal{L}(\mathcal{H}_{(BE)^n})$ whose action on Choi states is given by $\mathcal{S}_\ast^{(n)}$.
\end{prop}

Before proving Proposition~\ref{prop:Choi}, let us briefly explain why the existence of the random Stinespring superchannel is a direct consequence of the above statements.

\begin{cor} For any $n\geq 1$, there exist an encoding channel $\mathcal{E}$, a memory system $M$, and a decoding channel $\mathcal{D}$ such that, for any quantum channel $\Phi_{A\to B}$,
\bb\label{eq:9}
\mathcal{D}\circ\left(\Phi^{\otimes n}\otimes {\rm Id}_M\right)\circ \mathcal{E} 
= \EE{U_E}\!\left[
\big((\id_B\otimes U_E)V_{A\to BE}\big)^{\otimes n}
(\,\cdot\,)
\big(V_{A\to BE}^\dagger(\id_B\otimes U_E^\dagger)\big)^{\otimes n}
\right],
\ee
where the expectation value is taken over Haar-random unitaries $U_E$ acting on
$\mathcal{H}_E$, and $V_{A\to BE}$ is any fixed Stinespring isometry associated with $\Phi_{A\to B}$.
\end{cor}

\begin{proof}
Let $\mathcal{S}^{(n)}$ be the superchannel constructed in Proposition~\ref{prop:Choi}. 
Then, for any channel $\Phi_{A\to B}$, the Choi operator $J^{\mathcal{S}^{(n)}[\Phi^{\otimes n}]}_{(A'BE)^n}$ of $\mathcal{S}^{(n)}[\Phi^{\otimes n}]$ can be written in terms of the Choi operator $J^{\Phi^{\otimes n}}_{(A'B)^n}=(J^{\Phi}_{A'B})^{\otimes n}$ of $\Phi^{\otimes n}$, where $J^{\Phi}_{A'B} \coloneqq \Phi_{A\to B}(\Gamma_{AA'})$, as
\bb \label{eq:10}
J^{\mathcal{S}^{(n)}[\Phi^{\otimes n}]}_{(A'BE)^n} &= \mathcal{S}^{(n)}_\ast\big(J^{\Phi^{\otimes n}}_{(A'B)^n}\big)\\
&= \Lambda^{(n)}_{\rm purify}\left((J^{\Phi}_{A'B})^{\otimes n}\right)\\
&= \Lambda^{(n)}_{\rm purify}\left(\Tr_{E^n}\!\left[V_{A\to BE}^{\otimes n}\Gamma_{A'A}^{\otimes n}V_{A\to BE}^{\dagger\, \otimes n}\right]\right)\\
&\eqt{(i)}\EE{U_E}\left[\big((\id_B\otimes U_E)V_{A\to BE}\big)^{\otimes n}
\Gamma_{A'A}^{\otimes n}
\big(V_{A\to BE}^\dagger(\id_B\otimes U_E^\dagger)\big)^{\otimes n}\right],
\ee
where $\mathcal{S}_\ast^{(n)}$ was introduced in~\eqref{eq:S_ast}, and $V_{A\to BE}$ is any arbitrary Stinespring representation of $\Phi_{A\to B}$. In~(i) we have observed that $V_{A\to BE}\ket{\Gamma}_{A'A}$ is a legitimate purification of $J^{\Phi}_{A'B}$. The last identity follows from~\eqref{eq:property}. 
Now, by Choi's theorem,~\eqref{eq:10} immediately implies that, for any $\rho_{A^n}\in \mathcal{L}(\mathcal{H}_{A^n})$,
\bb
\mathcal{S}^{(n)}[\Phi^{\otimes n}](\rho_{A^n}) &= \Tr_{{A'}^n}\!\left[\big(\rho_{{A'}^n}^{\intercal}\otimes\id_{B^nE^n}\big)\, J^{\mathcal{S}[\Phi^{\otimes n}]}_{(A'BE)^n} \right] \\
&\eqt{(ii)} \EE{U_E} \Tr_{{A'}^n}\!\left[\big((\id_B\otimes U_E)V_{A\to BE}\big)^{\otimes n}\big(\id_{{A'}^n} \otimes \rho_{A^n} \big) \Gamma_{A'A}^{\otimes n} \big(V_{A\to BE}^\dagger(\id_B\otimes U_E^\dagger)\big)^{\otimes n}\right] \\
&\eqt{(iii)} \EE{U_E}\left[\big((\id_B\otimes U_E)V_{A\to BE}\big)^{\otimes n}\rho_{A^n} \big(V_{A\to BE}^\dagger(\id_B\otimes U_E^\dagger)\big)^{\otimes n}\right],
\ee
where in~(ii) we have commuted $\rho_{{A'}^n}^\intercal$ past the isometry and transferred its action on $A^n$ using the transpose trick, while in~(iii) we have observed that $\Tr_{A'}\Gamma_{A'A} = \id_A$. 

Since $\mathcal{S}$ is a superchannel, according to Lemma~\ref{lemma_giulio} it can be implemented by means of an encoder, a decoder and a memory. Hence~\eqref{eq:9} holds true, and this concludes the proof.
\end{proof}

We are only left with the proof of Proposition~\ref{prop:Choi}, which consists in a simple verification of the conditions~(a) and~(b) discussed at the beginning of this section.

\begin{proof}[Proof of Proposition~\ref{prop:Choi}.]
The map $\mathcal{S}^{(n)}_\ast$ is manifestly completely positive, as $\Lambda^{(n)}_{\rm purify}$ is completely positive. Now, let $X_{(AB)^n}$ be an operator on $\mathcal{H}_{A^n}\otimes \mathcal{H}_{B^n}$ such that $\Tr_{B^n} X_{(AB)^n} = \kappa \id_{A^n}$. Then,
\bb
\Tr_{B^nE^n} \mathcal{S}_\ast^{(n)}(X_{(AB)^n}) 
&\eqt{(i)}\Tr_{B^n}\Tr_{E^n}\Lambda^{(n)}_{\rm purify}\left(\mathcal{P}_{AB}^{(n)}\big(X_{(AB)^n}\big)\right)\\
&\eqt{(ii)}\Tr_{B^n}\left[\mathcal{P}_{AB}^{(n)}\big(X_{(AB)^n}\big)\right]\\
&=\mathcal{P}_{A}^{(n)}\left(\Tr_{B^n}\!\left[X_{(AB)^n}\right]\right)\\
&\eqt{(iii)}\kappa\id_{A^n},
\ee
where we have called $\mathcal{P}_{AB}^{(n)}$ and $\mathcal{P}_{A}^{(n)}$ the unital channels
\bb
    \mathcal{P}_{AB}^{(n)}(\,\cdot\,)&\coloneqq \frac{1}{n!}\sum_{\pi \in S_n} \big(P^{A^n}_{\pi}\otimes P^{B^n}_{\pi}\big) (\,\cdot\,) \big(P^{A^n}_{\pi}\otimes P^{B^n}_{\pi}\big)^\dagger,\\
    \mathcal{P}_{A}^{(n)}(\,\cdot\,)&\coloneqq \frac{1}{n!}\sum_{\pi \in S_n} P^{A^n}_{\pi} (\,\cdot\,)  P^{A^n\, ^\dagger}_{\pi},
\ee
respectively. In (i) we have used that the random purification channel symmetrises the input, in (ii) we have recalled that, for permutation invariant inputs, the output of the random purification reduces to the original state when tracing out the auxiliary system $E^n$, and in (iii) we have noticed that the unital channel $\mathcal{P}_{A}^{(n)}$ acts on  $\Tr_{B^n}\left[X_{(AB)^n}\right]=\kappa\id_{A^n}$. This concludes the proof.
\end{proof}

\section{Proof of Theorem~\ref{thm1} and quantum circuit for the random Stinespring superchannel}\label{explicit_circuit}

\subsection{Representation theory}
In this section, we provide a brief overview of the representation-theoretic tools required for our analysis. For a detailed introduction to these topics, we refer the readers to Refs.~\cite{hayashi_group_2017, Hayashi2016_grouptheoretic}.

Let $\mathcal{H}$ be a Hilbert space of dimension $d$. Let $\mathrm{U}(d)$ be the group of unitary matrices of size $d\times d$, and let $S_{n}$ be the symmetric group of $n$ elements. The space $\mathcal{H}^{\otimes n}$ hosts a representation of the group $\mathrm{U}(d)$ as the action of $U^{\otimes n}$, where $U\in \mathrm{U}(d)$, and a representation of $S_n$, as the action of permutations of the $n$ systems. We denote the permutation unitary corresponding to the permutation $\sigma$ as $U_{\sigma}$. The irreducible representations of $\mathrm{U}(d)$ and $S_n$ are labeled by \textit{Young diagrams}, i.e. ordered partitions $\lambda$ of $n$. The actions of the representations of $\mathrm{U}(d)$ and $S_n$ commute, and as a representation of $S_n\times \mathrm{U}(d) $ we have the following decomposition into irreducible representations of $\mathcal{H}^{\otimes n}$ (\textit{Schur--Weyl} duality~\cite{goodman_symmetry_2009,hayashi_group_2017}):

\begin{equation}\label{eq:SchurWeyl}
\mathcal{H}^{\otimes n}=\bigoplus_{\lambda\vdash n}[\lambda]\otimes \mathcal{U}_{d,\lambda}\,,
\end{equation}
where the sum runs over all lists of integers $\lambda=(\lambda_1,\cdots,\lambda_{l(\lambda)})$ with $l(\lambda)\leq d$ and \mbox{$\lambda_1+\cdots+\lambda_{l(\lambda)}=n$}. 
Moreover, $\mathcal{U}_{d,\lambda}$ is an irreducible representation of $\mathrm{U}(d)$ of dimension $\mathrm{dim}[\mathcal{U}_{d,\lambda}]$, and $[\lambda]$ is an irreducible representation of $S_n$ of dimension $\mathrm{dim}[\lambda]$. A preferred basis for each representation space $[\lambda]$ is the Young--Yamanouchi basis; we denote the associated matrix elements of the representation $\lambda$ evaluated on $\sigma\in S_n$ as $R_{\lambda}(\sigma)_{k,l}$. Note that, with this choice, the representation matrices are real-valued. The character of the representation $[\lambda]$ is denoted as $\chi_{\lambda}(\sigma)=\Tr[R_{\lambda}(\sigma)]$, and it is also clearly real-valued. For the unitary representation spaces $\mathcal{U}_{d,\lambda}$, a canonical choice is the Gelfand--Tsetlin basis. These two basis choices give rise to a basis of $\mathcal{H}^{\otimes n}$ that respects the structure of the decomposition of Eq.~\eqref{eq:SchurWeyl}: we write that basis as $\{\ket{\lambda,i,\alpha}\}_{\lambda\vdash n,\, i\in [\mathrm{dim[\lambda]}],\, \alpha\in{\mathrm{dim}[\mathcal{U}_{d,\lambda}]}}$.
The \textit{Schur transform} is the unitary operator that rotates this basis into the canonical basis~\cite{bacon_efficient_2006,harrow_applications_2005, krovi_efficient_2019, burchardt_high-dimensional_2025}. 
The isotypical projector $\Pi_{\lambda}$ on the subspace $ [\lambda]\otimes \mathcal{U}_{d,\lambda}$ in~\eqref{eq:SchurWeyl} can be written as
\begin{equation}\label{eq:isoproj}
\Pi_{\lambda}=\frac{\mathrm{dim}[\lambda]}{n!}\sum_{\sigma\in S_n}\chi_{\lambda}(\sigma)\, U_{\sigma}\,.
\end{equation}
We also recall that in any unitary representation $R:S_{n}\rightarrow \mathcal{L}(\mathcal{H})$ of the symmetric group, the operators $\Pi_{\lambda}^{\mathcal{H}}\coloneqq \frac{\mathrm{dim}[\lambda]}{n!}\sum_{\sigma\in S_n}\chi_{\lambda}(\sigma)R(\sigma)$ either project onto the subspace $[\lambda]\otimes M_{\lambda}$ where $M_{\lambda}$ is the multiplicity space of the irreducible representation $[\lambda]$, if $[\lambda]$ is in the decomposition into irreducible of $R$, or are equal to zero otherwise.
Together with standard representation theory tools, we will also use the machinery of Weingarten calculus~\cite{collins_weingarten_2022,kostenberger_weingarten_2021,harrow_approximate_2023,mele_introduction_2024}. 
The following expression for the $\mathrm{U}(d)$ twirl of an operator holds~\cite{mele_introduction_2024}:
\begin{equation}\label{eq:haarint}
\EE{U} \left[U^{\otimes n}A {U^{\dagger}}{}^{\otimes n}\right]=\sum_{\sigma,\tau\in S_n}\mathrm{Wg}(\tau^{-1}\sigma,d)\Tr[U^{\dagger}_{\sigma}A]\, U_{\tau},
\end{equation}
where $\EE{U}$ denotes the expectation value over the Haar measure on $\mathrm{U}(d)$, and $\mathrm{Wg}(\,\cdot\,)$ is the \textit{Weingarten function}, which can be computed as~\cite{collins_integration_2006} 
\begin{equation}\label{eq:weingchar}
\mathrm{Wg}(\sigma, d)=\frac{1}{(n!)^2}\sum_{\lambda\vdash n,l(\lambda)\leq d}\frac{\mathrm{dim}^2[\lambda]}{\mathrm{dim}[\mathcal{U}_{d,\lambda}]}\, \chi_{\lambda}(\sigma)\,.
\end{equation}
We will also need the quantum Fourier transform for $S_n$. Let  $\widehat{S_n}$ be a Hilbert space of dimension $n!$, and let $\{\ket{\sigma}\}_{\sigma \in S_n}$  be a basis of $\widehat{S_n}$. This space hosts the commuting left and right regular representations $\rho_L, \rho_R$ of the symmetric group, acting as $\rho_L(\sigma)\rho_{R}(\sigma'){\ket{\tau}}=\ket{\sigma\tau(\sigma')^{-1}}$. It is well known that, as a representation space for $S_{n}\times S_n$, the Hilbert space $\widehat{S_n}$ decomposes as 

\begin{equation}
\widehat{S_n}=\bigoplus_{\lambda\vdash n}[\lambda]\otimes [\lambda]\,.
\end{equation}
Therefore, a basis for $\widehat{S_n}$ is also given by $\{\ket{\lambda,i,j}\}_{\lambda\vdash n,\, i,j\in [\mathrm{dim}[\lambda]]}$. The map $\mathrm{QFT}$ is the unitary map~\cite{beals_quantum_1997,moore_symmetric_2008,kawano_quantum_2016} such that:
\begin{equation}\label{eq:qft}
\mathrm{QFT}\,\ket{\sigma}=\sum_{\substack{\lambda\vdash n\\ i,j\in [\mathrm{dim}[\lambda]]} }\sqrt{\frac{\mathrm{dim}[\lambda]}{n!}}\, R_{\lambda}(\sigma)_{i,j}\ket{\lambda,i,j}.
\end{equation}
For any $\pi\in S_n$ we also define the unitary $\mathsf{C\pi}$, known as controlled permutation unitary, acting on $\widehat{S_n}\otimes \mathcal{H}^{\otimes n}$ as
\begin{equation}
\mathsf{C\pi}\,(\ket{\sigma}\otimes \ket{\psi})=\ket{\sigma}\otimes U_{\sigma}\ket{\psi}\,.
\end{equation}
Moreover, its inverse acts as
\begin{equation}
(\mathsf{C\pi})^\dagger\,(\ket{\sigma}\otimes \ket{\psi})=\ket{\sigma}\otimes U^{\dagger}_{\sigma}\ket{\psi}\,.
\end{equation}

\subsection{An explicit formula for the random Stinespring isometry}

In this subsection, we derive an explicit expression for the random Stinespring isometry in terms of the channel $\Phi$ and permutation unitaries, using the representation theoretic tools and the notation introduced above.

\begin{lemma}[(Explicit formula for random Stinespring isometry)]\label{lemma_explicit_form}
Let $\Omega$ denote the channel induced by the random Stinespring isometry appearing
on the right-hand side of Eq.~\eqref{eq_random_isometry}, namely
\begin{equation}\label{def_superchannel}
\Omega(\,\cdot\,)\coloneqq 
\EE{U_E}\left[ \,(\id_{B}\otimes U_{E})^{\otimes n}V^{\otimes n}\left(\,\cdot\,\right) V^{\dagger}{}^{\otimes n} (\id_{B}\otimes U_{E}^{\dagger}){}^{\otimes n}\right]\,,
\end{equation}
where we recall that the dimension of the Stinespring environment $E$ is $r$. Then, this channel can also be expressed as:
\begin{equation}\label{eq:weingresult}
\Omega(\,\cdot\,)=\sum_{\sigma\in S_n}\,\sum_{\lambda\vdash n,\, l(\lambda)\leq r}\frac{1}{n!}\frac{\mathrm{dim}[\lambda]}{\mathrm{dim}[\mathcal{U}_{r,\lambda}]}\left(  U_{\sigma}\Phi^{\otimes n}(U^{\dagger}_{\sigma}(\,\cdot\,))\right)\otimes \left(U_{\sigma}\Pi_{\lambda}\right)\,.
\end{equation}
\end{lemma}

\begin{proof}
Without loss of generality, it is sufficient to verify~\eqref{eq:weingresult} for pure states as inputs. Using~\eqref{eq:haarint}, we have
\begin{equation}\label{eq:firststep}
\Omega(\ketbra{\psi})=\sum_{\sigma,\tau\in S_n} \mathrm{Wg}(\sigma^{-1}\tau)\Tr_{E^n}[(\id_{B^n}\otimes U_{\sigma}^{\dagger})V^{\otimes n}\left(\ketbra{\psi}\right) V^{\dagger}{}^{\otimes n}]\otimes U_{\tau}.
\end{equation}
Now, we note that
\bb
(\id_{B^n}\otimes U_{\sigma}^{\dagger})V^{\otimes n}\ket{\psi}&=(U_{\sigma}\otimes \id_{E^n})(U_{\sigma}^{\dagger}\otimes U_{\sigma}^{\dagger})V^{\otimes n}\ket{\psi} \\
&=(U_{\sigma}\otimes \id_{E^n})V^{\otimes n}(U_{\sigma}^{\dagger}\ket{\psi}) \,,
\ee
where we used that $V_{A\to BE}^{\otimes n}(U_{\sigma}^{\dagger}\ket{\psi}_{A^n})=(U_{\sigma}^{\dagger}\otimes U_{\sigma}^{\dagger})V_{A\to BE}^{\otimes n}\ket{\psi}_{A^n}$. Inserting this into~\eqref{eq:firststep}, we have
\bb\label{eq:secondstep}
\Omega(\ketbra{\psi}) &= \sum_{\sigma,\tau\in S_n} \mathrm{Wg}(\sigma^{-1}\tau) \Tr_{E^n}\!\left[(U_{\sigma}\otimes \id_{E^n})V^{\otimes n}\left(U_{\sigma}^{\dagger}\ketbra{\psi}\right) V^{\dagger}{}^{\otimes n}\right]\otimes U_{\tau}\\
&= \sum_{\sigma,\tau\in S_n} \mathrm{Wg}(\sigma^{-1}\tau) \left( U_{\sigma}\Phi^{\otimes n}(U_{\sigma}^{\dagger}\ketbra{\psi})\right)\otimes U_{\tau}\\
&\eqt{(i)} \sum_{\sigma,\tau\in S_n}\frac{1}{(n!)^2}\sum_{\lambda\vdash n,\, l(\lambda)\leq r}\frac{\mathrm{dim}^2[\lambda]}{\mathrm{dim}[\mathcal{U}_{r,\lambda}]}\, \chi_{\lambda}(\sigma^{-1}\tau) \left(U_{\sigma}\Phi^{\otimes n}(U_{\sigma}^{\dagger}\ketbra{\psi})\right)\otimes U_{\tau}\\
&\eqt{(ii)} \sum_{\sigma,\tau'\in S_n}\frac{1}{(n!)^2}\sum_{\lambda\vdash n,\, l(\lambda)\leq r}\frac{\mathrm{dim}^2[\lambda]}{\mathrm{dim}[\mathcal{U}_{r,\lambda}]}\,\chi_{\lambda}(\tau') \left(U_{\sigma}\Phi^{\otimes n}(U_{\sigma}^{\dagger}\ketbra{\psi})\right)\otimes \left(U_{\sigma}U_{\tau'}\right)\\
&\eqt{(iii)} \sum_{\sigma\in S_n}\frac{1}{n!}\sum_{\lambda\vdash n,\, l(\lambda)\leq r}\frac{\mathrm{dim}[\lambda]}{\mathrm{dim}[\mathcal{U}_{r,\lambda}]}\left( U_{\sigma}\Phi^{\otimes n}(U_{\sigma}^{\dagger}\ketbra{\psi})\right)\otimes  \left(U_{\sigma}\Pi_{\lambda}\right)\,.
\ee
where in (i) we used~\eqref{eq:weingchar}, in (ii) we have changed variable $\tau \to \tau'=\sigma^{-1}\tau$, which is a one-to-one mapping, and in (iii) we used~\eqref{eq:isoproj}.
\end{proof}

\subsection{An explicit circuit for the random Stinespring superchannel}
In this subsection, we present an explicit quantum circuit that implements the random Stinespring superchannel. A schematic representation of the circuit is
shown in Fig.~\ref{fig:figurecircuit}. We begin by introducing the individual
components that make up the circuit.

Let $\mathsf{E}:\mathcal{H}_A^{\otimes n}\to \widehat{S_n}\otimes
\mathcal{H}_A^{\otimes n}$ denote the (encoding) isometry defined by its action on
any state $\ket{\psi}\in\mathcal{H}_A^{\otimes n}$ as
\begin{equation}
\mathsf{E}\ket{\psi}
\coloneqq
\mathsf{C\pi}
\left[
\left(\frac{1}{\sqrt{n!}}\sum_{\sigma\in S_n}\ket{\sigma}\right)
\otimes
\ket{\psi}
\right],
\end{equation}
and let $\mathcal{E}(\,\cdot\,)\coloneqq \mathsf{E}(\cdot)\mathsf{E}^\dagger$ be
the corresponding isometry channel. Moreover, it is known that the uniform superposition over permutations can be efficiently prepared
as 
\bb
    \frac{1}{\sqrt{n!}}\sum_{\sigma\in S_n}\ket{\sigma}= \mathrm{QFT}^{\dagger}\ket{(n,0,\ldots,0),1,1}\,,
\ee
where the input state is
expressed in the basis
$\{\ket{\lambda,i,j}\}_{\lambda\vdash n,\; i,j\in[\mathrm{dim}[\lambda]]}$ of
$\widehat{S_n}$. Next, let
$\mathsf{D}:\widehat{S_n}\otimes \mathcal{H}_B^{\otimes n}\to
\widehat{S_n}\otimes \mathcal{H}_B^{\otimes n}$
denote the (decoding) unitary acting on any
$\ket{\phi}\in \widehat{S_n}\otimes \mathcal{H}_B^{\otimes n}$ as
\begin{equation}
\mathsf{D}\ket{\phi}
\coloneqq
(\mathrm{QFT}\otimes \id_{B^n})\, \mathsf{C\pi}^\dagger \ket{\phi}\,,
\end{equation}
and let $\mathcal{D}(\,\cdot\,)\coloneqq \mathsf{D}(\,\cdot\,)\mathsf{D}^\dagger$
be the associated (decoding) channel. Finally, we define a quantum channel
$\mathcal{T}:\mathcal{L}(\widehat{S_n})\to \mathcal{L}(\mathcal{H}_E^{\otimes n})$
by its action on the basis operators as

\begin{equation}\label{eq_tau}
\mathcal{T}(\ketbraa{\lambda,i,j}{\lambda',k,l})
\coloneqq
\delta_{\lambda,\lambda'}\delta_{i,k}\,
\ketbraa{\lambda,j}{\lambda,l}
\otimes
\frac{\id_{\mathcal{U}_{r,\lambda}}}{\mathrm{dim}[\mathcal{U}_{r,\lambda}]},
\end{equation}
if $l(\lambda)\leq r$, and as the replacer with the maximally mixed state over $E^n$ otherwise. Note that the channel $\mathcal{T}$ acts as a measurement of the index $\lambda$, and as a $\lambda$-dependent replacer channel on the subsystem hosting $i,k$: the overall action is depicted as $\mathcal{P}_{\lambda,r}$ in Figure~\ref{fig:figurecircuit}, and it consists in preparing the state $\ketbra{\lambda}\otimes \frac{\id_{\mathcal{U}_{r,\lambda}}}{\mathrm{dim}[\mathcal{U}_{r,\lambda}]}$ upon recording the outcome $\lambda$.

\begin{figure}[t]
  \centering
  \def\svgwidth{\linewidth}
  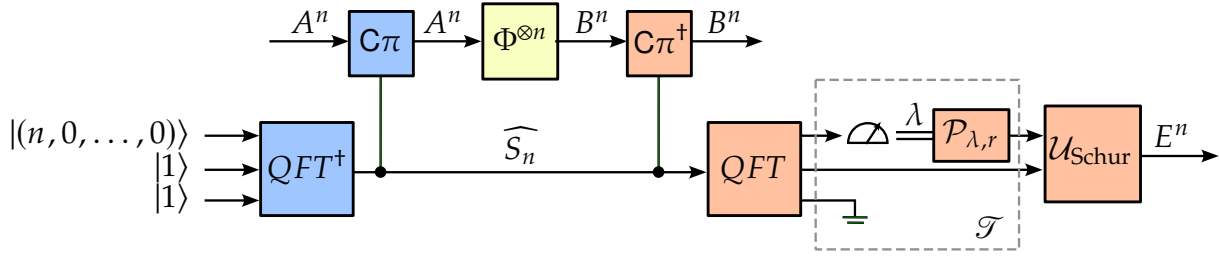
  \caption{Circuit implementation of the random Stinespring superchannel from
  Theorem~\ref{th:circuit}.}
  \label{fig:figurecircuit}
\end{figure}

With this notation in place, we are now ready to state and prove the main result of this subsection.
\begin{thm}[(Circuit implementing the random Stinespring superchannel)]\label{circuit_stin}
\label{th:circuit}
The quantum channel
\begin{equation}
\mathcal{C}_{A^n\to B^n E^n}
\coloneqq
(\mathcal{U}_{\mathrm{Schur}}\otimes \mathrm{Id}_{B^n})
\circ
\big(\mathcal{T}_{ \widehat{S_n}\to E^n }\otimes \mathrm{Id}_{B^n}\big)
\circ
\mathcal{D}
\circ
\big(\mathrm{Id}_{\widehat{S_n}}\otimes \Phi_{A\to B}^{\otimes n}\big)
\circ
\mathcal{E}_{A^n\to \widehat{S_n} \,A^n},
\end{equation}
which corresponds to the circuit depicted in Fig.~\ref{fig:figurecircuit}, is
exactly equal to the channel $\Omega_{A^n\to B^n E^n}$  induced by the random Stinespring isometry
as defined in Eq.~\eqref{def_superchannel}, namely
\bb
\mathcal{C}_{A^n\to B^n E^n}(\,\cdot\,)
=
\Omega_{A^n\to B^n E^n}(\,\cdot\,)\,,
\ee
where the equality is understood with the appropriate identification of the
subsystems in the tensor product.\footnote{In contrast to Lemma~\ref{lemma_explicit_form}, throughout this section we adopt,
for ease of presentation, the convention of writing the system $B^n$ on the
right-hand side of tensor products, while placing the auxiliary register
$\widehat{S_n}$ and the dilation environment $E^n$ on the left. This choice
simplifies the notation for controlled operations, for which it is natural to
display the control system on the left.}
\end{thm}

\begin{proof}
It suffices to prove that
\begin{equation}
\mathcal{C}(\ketbra{\psi})=\Omega(\ketbra{\psi})\,\qquad\forall \ket{\psi}\in \mathcal{H}_{A}^{\otimes n}\,.
\end{equation} 
Let us first compute the left-hand side:
\begin{align}
\mathcal{E}(\ketbra{\psi})=\frac{1}{n!}\sum_{\sigma,\tau\in S_n}\ketbraa{\sigma}{\tau}\otimes U_{\sigma}\ketbra{\psi}U^{\dagger}_{\tau}\,.
\end{align}
Then, applying $\Phi^{\otimes n}$ on the register $A^n$, we obtain
\bb
\big(\mathrm{Id}_{\widehat{S_n}}\otimes\Phi^{\otimes n}\big) \circ \mathcal{E}(\ketbra{\psi})&=\frac{1}{n!}\sum_{\sigma,\tau\in S_n}\ketbraa{\sigma}{\tau}\otimes \Phi^{\otimes n}(U_{\sigma}\ketbra{\psi}U^{\dagger}_{\tau})\\
&=\frac{1}{n!}\sum_{\sigma,\tau\in S_n}\ketbraa{\sigma}{\tau}\otimes U_{\tau}\Phi^{\otimes n}(U_{\tau^{-1}\sigma}\ketbra{\psi})U^{\dagger}_{\tau}\,\,,
\ee
where in the last line we exploited that $\Phi^{\otimes n} (\,\cdot\,)=U_\tau\Phi^{\otimes n}(\,U_\tau^\dagger (\,\cdot\,)U_\tau)U_\tau^\dagger$.
Going forward, let us now apply $\mathcal{D}$:
\bb
&\mathcal{D}\circ \big(\mathrm{Id}_{\widehat{S_n}}\otimes \Phi^{\otimes n}\big)\circ \mathcal{E}(\ketbra{\psi})\\
&\qquad =\frac{1}{n!}\sum_{\sigma,\tau\in S_n}\sum_{\lambda,\lambda',i,j,k,l}\frac{\sqrt{\mathrm{dim}[\lambda]\mathrm{dim}[\lambda']}}{n!}\!\! && R_{\lambda}(\sigma)_{i,j}\ketbraa{\lambda,i,j}{\lambda',k,l}R_{\lambda'}(\tau)_{k,l} \,\otimes \\[-0.6em]
& &&  \quad\otimes U_{\sigma^{-1}\tau}\Phi^{\otimes n}(U_{\tau^{-1}\sigma}\ketbra{\psi})\,,
\ee
where we used~\eqref{eq:qft} and the fact that representation matrices $R_{\lambda}(\sigma)$ are real-valued.

Now, we would like to apply the channel $\mathcal{T}$ to the auxiliary register $\widehat{S_n}$.
Recall that $\mathcal{T}$, defined in Eq.~\eqref{eq_tau}, acts differently
depending on whether the associated Young diagram $\lambda$ satisfies
$l(\lambda)\le r$ or not. In order to simplify the analysis, we first observe that only
the terms with $l(\lambda)\le r$ give a nonzero contribution. Indeed, note that
\bb
\sum_{\sigma,\tau\in S_n}&\sum_{\lambda,i,j,l}\frac{\mathrm{dim}[\lambda]}{n!^2}R_{\lambda}(\sigma)_{i,j}R_{\lambda}(\tau)_{i,l}\ketbraa{\lambda,j}{\lambda,l}\otimes U_{\sigma^{-1}\tau}\Phi^{\otimes n}(U_{\tau^{-1}\sigma}\ketbra{\psi})\\
&\eqt{(v)}\sum_{\sigma,\tau\in S_n}\sum_{\lambda,i,j,l}\frac{\mathrm{dim}[\lambda]}{n!^2}R_{\lambda}(\sigma^{-1})_{j,i}R_{\lambda}(\tau)_{i,l}\ketbraa{\lambda,j}{\lambda,l}\otimes U_{\sigma^{-1}\tau}\Phi^{\otimes n}(U_{\tau^{-1}\sigma}\ketbra{\psi})\\
&\eqt{(vi)}\sum_{\sigma\in S_n}\sum_{\lambda,j,l}\frac{\mathrm{dim}[\lambda]}{n!}R_{\lambda}(\sigma)_{j,l}\ketbraa{\lambda,j}{\lambda,l}\otimes U_{\sigma}\Phi^{\otimes n}(U_{\sigma^{-1}}\ketbra{\psi})\\
&\eqt{(vii)}\sum_{\sigma\in S_n}\sum_{\lambda,j,l}\frac{\mathrm{dim}[\lambda]}{n!}R_{\lambda}(\sigma)_{j,l}\ketbraa{\lambda,j}{\lambda,l}\otimes U_{\sigma}\Tr_{E^n}[V^{\otimes n}(U_{\sigma^{-1}}\ketbra{\psi})V^{\dagger}{}^{\otimes n}]\\
&\eqt{(viii)}\sum_{\sigma\in S_n}\sum_{\lambda,j,l}\frac{\mathrm{dim}[\lambda]}{n!}R_{\lambda}(\sigma)_{j,l}\ketbraa{\lambda,j}{\lambda,l}\otimes \Tr_{E^n}[(\id_{B^n}\otimes U_{\sigma^{-1}}) V^{\otimes n}(\ketbra{\psi})V^{\dagger}{}^{\otimes n}],
\ee
where in (v) we used that the $R_{\lambda}(\sigma)^{\dagger}=R_{\lambda}(\sigma)^{\intercal}=R_{\lambda}(\sigma^{-1})$ are real, and in (vi) we contracted the two representation matrices $R_{\lambda}$ and changed variable, in (vii) we wrote $\Phi$ in terms of its dilation, and in (viii) we used the permutation covariance of $V^{\otimes n}$. By using that the dimension of the Stinespring environment $E$ is $r$, we can now insert the resolution of the identity in the $E^n$ space 
\bb
    \id_{E^n}=\sum_{\lambda'\vdash n, l(\lambda')\leq r}\Pi_{\lambda'}=\sum_{\lambda'\vdash n, l(\lambda')\leq r}\frac{\mathrm{dim}[\lambda']}{n!}\sum_{\pi\in S_n}\chi_{\lambda'}(\pi)U_{\pi}\,,
\ee
and use the invariance of the measure on the group to obtain
\bb
&\sum_{\sigma\in S_n}\sum_{\lambda,j,l}R_{\lambda}(\sigma)_{j,l}\ketbraa{\lambda,j}{\lambda,l}\otimes \Tr_{E^n}[(\id_{B^n}\otimes U_{\sigma^{-1}}) V^{\otimes n}(\ketbra{\psi})V^{\dagger}{}^{\otimes n}]\\
&\qquad \eqt{(ix)}\sum_{\lambda'\vdash n,\, l(\lambda')\leq r}\sum_{\sigma,\pi\in S_n}\sum_{\lambda,j,l}\frac{\mathrm{dim}[\lambda']}{n!}\chi_{\lambda'}(\pi)R_{\lambda}(\sigma)_{j,l}\ketbraa{\lambda,j}{\lambda,l}\\&\qquad\otimes \Tr_{E^n}[(\id_{B^n}\otimes U_{\sigma^{-1}}U_{\pi}) V^{\otimes n}(\ketbra{\psi})V^{\dagger}{}^{\otimes n}]\\
&\qquad \eqt{(x)}\sum_{\lambda'\vdash n,\, l(\lambda')\leq r}\sum_{\sigma,\pi\in S_n}\sum_{\lambda,j,l}\frac{\mathrm{dim}[\lambda']}{n!}\chi_{\lambda'}(\pi)R_{\lambda}(\sigma\pi)_{j,l}\ketbraa{\lambda,j}{\lambda,l}\\&\qquad\otimes \Tr_{E^n}[(\id_{B^n}\otimes U_{\sigma^{-1}}) V^{\otimes n}(\ketbra{\psi})V^{\dagger}{}^{\otimes n}]\\
&\qquad \eqt{(xi)}\sum_{\lambda'\vdash n,\, l(\lambda')\leq r}\sum_{\sigma\in S_n}\sum_{\lambda,j,l}R_{\lambda}(\sigma)_{j,l}\Pi_{\lambda'}^{\widehat{S}_n}\ketbraa{\lambda,j}{\lambda,l}\otimes \Tr_{E^n}[(\id_{B^n}\otimes U_{\sigma^{-1}}) V^{\otimes n}(\ketbra{\psi})V^{\dagger}{}^{\otimes n}]\\
&\qquad \eqt{(xii)}\sum_{\lambda\vdash n,\, l(\lambda)\leq r}\sum_{\sigma\in S_n}\sum_{j,l}R_{\lambda}(\sigma)_{j,l}\ketbraa{\lambda,j}{\lambda,l}\otimes \Tr_{E^n}[(\id_{B^n}\otimes U_{\sigma^{-1}}) V^{\otimes n}(\ketbra{\psi})V^{\dagger}{}^{\otimes n}]\\
&\qquad=\sum_{\lambda\vdash n,\, l(\lambda)\leq r}\sum_{\sigma\in S_n}\sum_{j,l}R_{\lambda}(\sigma)_{j,l}\ketbraa{\lambda,j}{\lambda,l}\otimes U_{\sigma}\Phi^{\otimes n}(U_{\sigma^{-1}}\ketbra{\psi})\,,
\ee
where in~(ix) we inserted the resolution of the identity in terms of $\Pi_{\lambda}$ and their expression as linear combinations of permutations, in~(x) we used the invariance of the measure on the group, in~(xi) we recollected the isotypical projector $\Pi_{\lambda'}^{\widehat{S}_n}$ in the representation space $\bigoplus_{\lambda \vdash n}[\lambda]$ and in~(xii) we used that $\Pi_{\lambda'}^{\widehat{S}_n}\ket{\lambda,i}=\delta_{\lambda,\lambda'}\ket{\lambda,i}$. This means that when $V$ is defined with an environment of dimension at most $r$, we can restrict the sum over Young diagrams with length at most $r$.

Then, applying $\mathcal{T}$ to the register $\widehat{S_n}$, we obtain:
\bb
&(\mathcal{T}\otimes \mathrm{Id}_{B^n})\circ \mathcal{D} \circ \big(\mathrm{Id}_{\widehat{S_n}}\otimes \Phi^{\otimes n}\big) \circ \mathcal{E}(\ketbra{\psi}) \\
&\qquad=\sum_{\sigma\in S_n}\sum_{\lambda\vdash n,\, l(\lambda)\leq r}\frac{\mathrm{dim}[\lambda]}{n!\mathrm{dim}[\mathcal{U}_{r,\lambda}]}\sum_{j,l}R_{\lambda}(\sigma)_{j,l}\ketbraa{\lambda,j}{\lambda,l}\otimes\id_{\mathcal{U}_{r,\lambda}}  \otimes U_{\sigma}\Phi^{\otimes n}(U_{\sigma}^{\dagger}\ketbra{\psi}) 
\ee
As a consequence, we conclude that
\bb
\mathcal{C}(\ketbra{\psi})&=
\sum_{\sigma\in S_n}\sum_{\lambda\vdash n,\, l(\lambda)\leq r}\frac{\mathrm{dim}[\lambda]}{n!\mathrm{dim}[\mathcal{U}_{r,\lambda}]}\, \Pi_{\lambda}U_{\sigma}\otimes U_{\sigma}\Phi^{\otimes n}(U_{\sigma}^{\dagger}\ketbra{\psi})\,,
\ee
where we used the Schur--Weyl duality in~\eqref{eq:SchurWeyl} to write that 
\bb
U_{\mathrm{Schur}}^{\dagger}\Pi_{\lambda}U_{\sigma}U_{\mathrm{Schur}}= \sum_{j,l}R_{\lambda}(\sigma)_{j,l}\ketbraa{\lambda,j}{\lambda,l} \otimes \id_{\mathcal{U}_{r,\lambda}}\,.
\ee 
Comparing with~\eqref{eq:weingresult} (with the appropriate identification of the
subsystems in the tensor product), we obtain the claim.
\end{proof}

\begin{rem}[(Efficiency of the circuit)]\label{rema}
The circuit described in Theorem~\ref{th:circuit} has depth $O({\rm poly}(n,\log d,\log\frac{1}{\eta}))$, with $\eta$ being the diamond norm error due to finite gate set approximations. The circuit of $\mathrm{QFT}$ can be implemented in $\mathrm{poly}(n)$ time (see~\cite{beals_quantum_1997,moore_symmetric_2008} and a refined analysis in~\cite{kawano_quantum_2016}). In the implementation of~\cite{kawano_quantum_2016} the permutations are arranged in the memory through their canonical encoding $\sigma=(c_{1,...,n})^{i_{n-1}}(c_{1,...,n-1})^{i_{n-2}}\ldots (c_{1,2})^{i_1}$, where $c_{1,..,k}$ is the cycle on $\{1,...,k\}$ and $i_{k}\in\{0,\ldots,k\}$ for any $k\in \{2,\ldots,n\}$, so that $i_1$ can be stored in a qubit, $i_2$ in a qutrit, and so on until $i_{n-1}$, which is stored in an $n$-dimensional space. By applying cycles controlled by each of these registers in sequence, one can implement $\mathsf{C\pi}$ in time $O(\mathrm{poly}(n,\log d))$.  The channel $\mathcal{T}$ is just a partial trace composed with a preparation of a maximally mixed state: this preparation is depicted as $\mathcal{P}_{\lambda,r}$ in Figure~\ref{fig:figurecircuit} and it consists in preparing the state $\ketbra{\lambda}\otimes \frac{\id_{\mathcal{U}_{r,\lambda}}}{\mathrm{dim}[\mathcal{U}_{r,\lambda}]}$. Finally, the last step is a Schur transform: circuits with complexity polynomial in the number of copies and dimension were first proposed in~\cite{bacon_efficient_2006}, and Harrow~\cite{harrow_applications_2005} sketched a method in his thesis to lower the dimension dependence to $O(\log d)$. A detailed proposal to achieve this was presented in~\cite{krovi_efficient_2019}, which was recently found to contain a mistake~\cite{Fei2024QuantumAlgorithm}. A corrected version of this proposal and an improved version of the original algorithm~\cite{bacon_efficient_2006} have been established by~\cite{burchardt_high-dimensional_2025}, confirming the $O(\mathrm{poly}(n,\log d))$ complexity. As a final note, the preparation $\ketbra{\lambda}\otimes \frac{\id_{\mathcal{U}_{r,\lambda}}}{\mathrm{dim}[\mathcal{U}_{r,\lambda}]}$ can of course be done classically, requiring sampling semi-standard Young tableaux of shape $\lambda$ uniformly and then preparing the corresponding basis state. However, this operation can be costly in terms of classical computation\footnote{We thank Elias Theil for noticing this issue.}. For a dilation dimension $r=2^k$, $k$ integer, a shortcut is to prepare the mixed state approximately (and compatibly with the diamond norm error already accounted for by the approximations in the Schur transform), via the following steps: 
\begin{enumerate}[(i)]
    \item Prepare the state $\ket{\lambda}\otimes\ket{S_0}\otimes \ket{T_0}$, where $\ket{S_0}$ is some Gelfand-Tsetlin (GT) pattern of shape $\lambda$, and $\ket{T_0}$ a valid Young-Yamanouchi basis state of shape $\lambda$ (this can be done efficiently: valid patterns can be computed and prepared in $O(\mathrm{poly}(n,\log d))$, see the encodings in~\cite{burchardt_high-dimensional_2025});
    \item Apply inverse Schur transform;
    \item apply $n$ copies of a random circuit approximating an $n$-design with error $\epsilon$ in diamond-norm with $O(\log (k/\epsilon)·n\mathrm{poly}\log(n))$ depth~\cite{Schuster2025};
    \item apply the Schur transform;
    \item Discard the permutation register.
\end{enumerate}
This procedure works because using the Haar measure instead of the approximate design one would obtain the $U(r)$  twirling of a GT pattern, which prepares the maximally mixed state in the irrep $\lambda$, and such twirling involves the $n$-th moment of the Haar measure.
\end{rem}

\begin{rem}[(Intuition behind the circuit)]
The reader may wonder how we came up with the circuit, and whether there is some intuition behind it. The process involved some trial and error to reproduce the desired supermap for $n=1$ and $n=2$ using controlled permutations, and the general ansatz was found by observing that the QFT method for weak Schur sampling based on the QFT in Chapter 8 of Harrow's thesis~\cite{harrow_applications_2005} was, in fact, implementing the main step of the random purification channel of~\cite{tang2025}. A more systematic, representation-theoretic derivation of the random Stinespring superchannel will be presented in a future version of this manuscript.
\end{rem}

We are now ready to prove our main result, stated in Theorem~\ref{thm1}.
\begin{proof}[Proof of Theorem~\ref{thm1}]
The existence of encoding maps, a memory system, and decoding maps satisfying~\eqref{eq_random_isometry} follows directly from Theorem~\ref{circuit_stin}. Furthermore, the efficiency of the circuit implementing the random Stinespring superchannel is established by Remark~\ref{rema}.
\end{proof}

\section{Applications to quantum Shannon theory}\label{appl_Shannon}


Throughout this section, we present an application of the random Stinespring superchannel to quantum Shannon theory, namely, the extension to quantum channels of Uhlmann-type theorems that are currently known only at the level of quantum states~\cite{NC,Mazzola_2025,Fang2025-variational,random_pur_simple}.

The celebrated Uhlmann theorem for fidelity is a fundamental result in quantum information theory. It states that the fidelity between two quantum states can be expressed as the maximum fidelity between their purifications~\cite{NC}. This theorem has been extended to more general quantum divergences~\cite{Mazzola_2025,Fang2025-variational}, and simpler proofs of these extensions have recently been obtained using the random purification channel~\cite{random_pur_simple}. Here, we use the random Stinespring superchannel to extend the Uhlmann theorem for quantum divergences (Theorem~\ref{thm:Uhlmann})~\cite{Mazzola_2025,Fang2025-variational,random_pur_simple} to the setting of quantum channels, thereby obtaining Theorem~\ref{thm:Uhlmann2}.


We recall that a function $\mathbb{D}:\mathcal{D}(\mathcal{H})\times\mathcal{D}(\mathcal{H})\rightarrow \mathbb{R}\cup\{+\infty\}$ is called \emph{divergence} if it satisfies the \emph{data-processing inequality}: for every quantum channel $\Lambda$ and every pair of states $(\rho,\sigma)$, we have
\bb
\mathbb{D}\big(\Lambda(\rho)\big\|\Lambda(\sigma)\big)\leq \mathbb{D}(\rho\|\sigma).
\ee
A divergence is \textit{jointly convex} if for any pair of ensembles of states $\{(p_i,\rho_i)\}_i$, $\{(p_i,\sigma_i)\}_i$ we have
\bb
\Rel{\mathbb{D}}{\sumno_{i}p_i\rho_i}{\sumno_{i}p_i\sigma_i} \leq \sum_{i}p_i \mathbb{D}(\rho_i\|\sigma_i).
\ee
Joint convexity is actually a consequence of the data-processing inequality whenever 
\bb
\Rel{\mathbb{D}}{\sumno_{i}p_i\ketbra{i}\otimes\rho_i}{\sumno_{i}p_i\ketbra{i}\otimes\sigma_i} = \sum_{i}p_i \mathbb{D}(\rho_i\|\sigma_i),
\ee
which holds for most divergences of interest.
Given any arbitrary divergence $\mathbb{D}$ between states, we can define a corresponding notion of divergence between channels: given quantum channels $\pazocal{M}$ and $\pazocal{N}$ with input system $A$ and output system $B$, we set
\bb \label{eq:rel_ent_ch}
\mathbb{D}\big(\pazocal{M}\,\big\|\,\pazocal{N}\big)\coloneqq \sup_{\rho_{RA}}\mathbb{D}\big(({\rm Id}_R\otimes \pazocal{M}_{A\to B})(\rho_{RA})\,\big\|\,({\rm Id}_R\otimes \pazocal{N}_{A\to B})(\rho_{RA})\big),
\ee
where the maximum is taken over all possible auxiliary systems $R$ and states $\rho_{RA}\in\mathcal{D}(\mathcal{H}_{RA})$.
One of the most relevant examples of divergence between states is the \emph{Umegaki relative entropy}~\cite{Umegaki1962}, defined as
\begin{equation} \label{eq:umegaki}
    D(\rho\|\sigma)
    \coloneqq \Tr\!\left[\rho\bigl(\log\rho - \log\sigma\bigr)\right],
\end{equation}
and which can be lifted to channels as
\bb
D\big(\pazocal{M}\,\big\|\,\pazocal{N}\big)\coloneqq \sup_{\rho_{RA}}D\big(({\rm Id}_R\otimes \pazocal{M}_{A\to B})(\rho_{RA})\,\big\|\,({\rm Id}_R\otimes \pazocal{N}_{A\to B})(\rho_{RA})\big),
\ee
Note that, since a divergence between states satisfies the data-processing inequality, also the corresponding version for channels satisfies an analogous data-processing inequality:
\bb
    D\big(\Lambda\circ \pazocal{M}\,\big\|\,\Lambda\circ\pazocal{N}\big)\leq 
    D\big(\pazocal{M}\,\big\|\,\pazocal{N}\big).
\ee
A divergence $\mathbb{D}$ between states is said to be \emph{additive} if, given two arbitrary Hilbert spaces $\mathcal{H}_1$ and $\mathcal{H}_2$, we have
\bb
    \mathbb{D}(\rho_1\otimes\rho_2\|\sigma_1\otimes\sigma_2)=\mathbb{D}(\rho_1\|\sigma_1)+\mathbb{D}(\rho_2\|\sigma_2)
\ee
for all states $\rho_1,\sigma_1\in\mathcal{H}_1$ and $\rho_2,\sigma_2\in\mathcal{H}_2$. This is clearly the case for the Umegaki relative entropy, due to the additivity of the matrix logarithm under tensor products. However, not all the divergences are additive (for example, the measured relative entropy is not~\cite{Donald1986}). In that case, we can also introduce the notion of \emph{regularisation} of $\mathbb D$, by setting 
\bb
\mathbb{D}^\infty(\rho\|\sigma)&\coloneqq \liminf_{n\to\infty}\frac{1}{n}\,\mathbb{D}(\rho^{\otimes n}\|\sigma^{\otimes n})\, . \\
\ee
Since most useful quantum divergences are either subadditive or superadditive, Fekete's lemma guarantees that for such divergences the above liminf is actually a limit. An additive divergence $\mathbb{D}$ between states might give rise to a non-additive notion of divergence between channels, due to the presence of entanglement at the input. This happens even in the simple case of the Umegaki relative entropy\cite[Proposition~3.1]{Fang2020}: indeed, there exist two channels $\pazocal{M}$ and $\pazocal{N}$ such that
\bb
D(\pazocal{M}\otimes\pazocal{M}\|\pazocal{N}\otimes\pazocal{N})>2D(\pazocal{M}\|\pazocal{N}).
\ee
It is then relevant to introduce the regularised version of~\eqref{eq:rel_ent_ch}:
\bb
\mathbb{D}^\infty(\pazocal{M}\|\pazocal{N})&\coloneqq \liminf_{n\to\infty}\frac{1}{n}\,\mathbb{D}(\pazocal{M}^{\otimes n}\|\pazocal{N}^{\otimes n})\,. \\
\ee
The Umegaki relative entropy is known to be \emph{weakly concave}: namely, for any ensemble of states $\{(p_i,\rho_i)\}_i$, we have
\begin{equation}\label{eq:weak_conc}
    D\!\left(\sum_{i=1}^N p_i \rho_i \,\middle\|\, \sigma\right)
    \ge \sum_{i=1}^N p_i\, D(\rho_i\|\sigma) + \sum_i p_i \log p_i
\end{equation}
The previous inequality can be weakened as
$D\!\left(\sum_{i=1}^N p_i \rho_i \,\middle\|\, \sigma\right)
    \ge \displaystyle{\min_{1\leq i\leq N}}\, D(\rho_i\|\sigma) -\log N$.
We say that an arbitrary divergence $\mathbb D$ is \emph{weakly quasi-concave} if an analogous property holds, namely if there exists a polynomial $P$, such that, for any $n\geq 1$, for any finite ensemble of states $\{(p_i,\rho_i)\}_{i=1,\dots, N}$ on an arbitrary Hilbert space $\mathcal{H}^{\otimes n}$, $\mathrm{dim}\,\mathcal{H}=d$, and for any state $\sigma\in\mathcal{D}(\mathcal{H}^{\otimes n})$, we have
\bb\label{eq_weak_quasi_conc}
    \mathbb{D}\left(\sum_{i=1}^N p_i\rho_i\,\middle\|\,\sigma\right)\geq \min_{1\leq i\leq N}\mathbb{D}(\rho_i\|\sigma)-\log P_d(N, s_\sigma),
\ee
where $s_\sigma\coloneqq |{\rm spec}(\sigma)|$.
%
Besides the Umegaki relative entropy, an important family of quantum divergences satisfies weak quasi-concavity: the sandwiched R\'enyi divergences $\tilde{D}_\alpha$ of order $\alpha\in [1/2,\infty]$~\cite{tomamichel12smooth_tutorial, newRenyi, Wilde2014} (see e.g.~\cite{random_pur_simple} for a concise proof).

In order to state the Uhlmann theorem for divergences between states found in~\cite{Mazzola_2025, Fang2025-variational,random_pur_simple}, we need a final definition.

\begin{Def}\label{def:ext}
    Given a state $\sigma_A\in\mathcal{D}(\mathcal{H}_A)$ and a Hilbert space $\mathcal{H}_B$ isomorphic to $\mathcal{H}_A$, we define the set $\pazocal{C}_{AB}^{\sigma_A}$ of $B$\emph{-extensions} of $\sigma_A$ as
    \bb
        \pazocal{C}_{AB}^{\sigma_A}\coloneqq\left\{\tilde \sigma_{AB}\in\mathcal{D}(\mathcal{H}_A\otimes\mathcal{H}_B)\,:\, \Tr_B\tilde\sigma_{AB}=\sigma_A\right\},
    \ee
    and the family $\mathcal{C}_{AB}^{\sigma_A}$ as the sequence $\mathcal{C}_{AB}^{\sigma_A}\coloneqq \left(\pazocal{C}_{A^nB^n}^{\sigma_A^{\otimes n}}\right)_{n\geq 1}$.
 According to standard conventions, the regularised relative entropy between an extension $\rho_{AB}$ of $\rho_A$ and the family $\mathcal{C}_{AB}^{\sigma_A}$ is then defined as
\bb
\rel{\mathbb{D}^\infty}{\rho_{AB}}{\mathcal{C}_{AB}^{\sigma_A}} &\coloneqq \liminf_{n\to\infty}\frac{1}{n}\inf_{\sigma_{A^nB^n}\in \pazocal{C}_{A^nB^n}^{\sigma_A^{\otimes n}}} \rel{\mathbb{D}}{\rho_{AB}^{\otimes n}}{\sigma_{A^nB^n}}\,.
\ee
\end{Def}

\begin{thm}[(Axiomatic Uhlmann's theorem for states~\cite{Mazzola_2025, Fang2025-variational,random_pur_simple})]\label{thm:Uhlmann}
Let $\mathbb{D}(\,\cdot\,\|\,\cdot\,)$ be a divergence satisfying weak quasi-concavity, i.e.~\eqref{eq_weak_quasi_conc}.
Then, given $\rho_A$ and $\sigma_A$ in $\mathcal{D}(\mathcal{H}_A)$, for any arbitrary extension $\rho_{AB}$ of $\rho_A$ we have
\bb\label{eq:Uhlmann}
    \mathbb{D}^{\infty}(\rho_A\|\sigma_A)= \rel{\mathbb{D}^\infty}{\rho_{AB}}{\mathcal{C}^{\sigma_A}_{AB}}\,.
\ee
\end{thm}
{}

Similarly to Definition~\ref{def:ext}, let us introduce the set of all extensions of a channel.
\begin{Def}
    Let $\mathcal{H}_A,\mathcal{H}_B$ and $\mathcal{H}_E$ be Hilbert spaces. Given a quantum channel $\pazocal{N}_{A\to B}$, an $E$\emph{-dilation} of $\pazocal{N}_{A\to B}$ is a quantum channel $\widebar{\pazocal{N}}_{A\to BE}$ such that
    \bb
        ({\rm Id}_B\otimes \Tr_E)\circ\widebar{\pazocal{N}}_{A\to BE}=\pazocal{N}_{A\to B}\,.
    \ee    
    We denote by $\pazocal{C}_{A\to BE}^{\pazocal{N}}$ the set of all the $E$-extensions of $\pazocal{N}_{A\to B}$, and we define the family $\mathcal{C}_{A\to BE}^{\pazocal{N}}$ to be the sequence $\mathcal{C}_{A\to BE}^{\pazocal{N}}\coloneqq \left(\pazocal{C}_{A^n\to B^n E^n}^{\pazocal{N}^{\otimes n}}\right)_{n\geq 1}$.
 Then, the regularised relative entropy between an extension $\widebar{\pazocal{M}}_{A\to BE}$ of $\pazocal{M}_{A\to B}$ and the family $\mathcal{C}_{A\to BE}^{\pazocal{N}}$ is defined as
\bb
\rel{\mathbb{D}^\infty}{\widebar{\pazocal{M}}_{A\to BE}}{\mathcal{C}_{A\to BE}^{\pazocal{N}}} &\coloneqq \liminf_{n\to\infty}\frac{1}{n}\inf_{\widebar{\pazocal{N}}\,\in\, \pazocal{C}_{A^n\to B^n E^n}^{\pazocal{N}^{\otimes n}}} \Rel{\mathbb{D}}{\widebar{\pazocal{M}}_{A\to BE}^{\otimes n}}{\widebar{\pazocal{N}}_{A^n\to B^nE^n}}\, .
\ee
\end{Def}

Finally, we need three technical lemmas in order to prove the main result of this section. The first one provides an extension of~\cite[Lemma~3]{random_pur_simple} to channels; the second one, instead, generalises the known fact that all the extensions of a quantum state can be obtained by applying a suitable channel to a fixed purification (see e.g.\ the discussion in~\cite[Section~III]{squashed}).

\begin{lemma}\label{lem:weak_qc}
    Let $\mathcal{H}_A$ and $\mathcal{H}_B$ be two Hilbert spaces, and let $\mathbb{D}$ be a weakly quasi-concave divergence according to~\eqref{eq_weak_quasi_conc}. Then,
\bb
    \mathbb{D}\left(\EE{\Lambda\sim \nu}\Lambda^{\otimes n}\,\middle\|\,\Gamma^{(n)}\right)\geq \sup_{\rho_{RA^n}}\min_{\Lambda\in\supp(\nu)} \Rel{\mathbb{D}}{\Lambda^{\otimes n}(\rho)}{\Gamma^{(n)}(\rho)} - \log {\rm poly}_{d}\Big(n,\,\big|\mathrm{spec}\big(\Gamma^{(n)}(\rho)\big)\big|\Big)
\ee
for all probability measures $\nu$ on the set of channels from $\mathcal{H}_A$ to $\mathcal{H}_B$, and all channels $\Gamma^{(n)}$ from $\mathcal{H}_A^{\otimes n}$ to $\mathcal{H}_B^{\otimes n}$, where $d\coloneqq(\dim\mathcal{H}_A)(\dim\mathcal{H}_B)$. 
\end{lemma}

\begin{proof}
Let $J_{(A'B)^n}$ be the Choi operator of the channel $\EE{\Lambda\sim \nu}\Lambda^{\otimes n}$, and let $H_{d,n}^{\rm sym}$ be the real vector space of permutationally symmetric Hermitian operators on $\mathcal{H}_{A'B}^{\otimes n}$; then,
\bb\label{eq:convex}
    J_{(A'B)^n} = \EE{\Lambda\sim \nu}\big[\Lambda^{\otimes n}_{A\to B}(\Gamma_{A'A}^{\otimes n})\big]= \EE{\Lambda\sim \nu}\big[(J^{\Lambda}_{A'B})^{\otimes n}\big]\in H_{d,n}^{\rm sym} .
\ee
By Schur--Weyl duality, $H_{d,n}^{\rm sym}$ has the form
\bb
H_{d,n}^{\rm sym}=\bigoplus_{\lambda\in\pazocal{Y}_n^d} \id_{[\lambda]}\otimes H(\pazocal{U}_\lambda),
\label{eq:space_perm_symm_operators}
\ee
where $\lambda$ ranges on the set $\pazocal{Y}_d^n$ of Young diagrams with size $n$ and depth at most $d$, $\pazocal{U}_\lambda$ and $[\lambda]$ are irreducible representations of the special unitary group ${\rm SU}(d)$ and of the symmetric group $S_n$, respectively, and $H(\pazocal{U}_\lambda)$ is the space of Hermitian operators on $\pazocal{U}_\lambda$. Leveraging the fact that $\dim \pazocal{U}_\lambda\leq (n+1)^{d(d-1)/2}$ and $|\pazocal{Y}_n^d|\leq (n+1)^{d-1}$~\cite[Eq.~(6.16) and~(6.18)]{Hayashi2016_grouptheoretic}, we can upper bound $\dim H_{d,n}^{\rm sym}\leq (n+1)^{d^2-1}$. As a consequence, since $J_{(A'B)^n}\in H_{d,n}^{\rm sym}$ belongs to the convex hull of $\big\{(J^{\Lambda}_{A'B})^{\otimes n}: \Lambda\in{\rm supp}(\nu) \big\}$, by Carath\'eodory's theorem we can write it as a convex combination of at most $N=(n+1)^{d^2-1} +1$ Choi operators $(J_{A'B}^{\Lambda_i})^{\otimes n}$ for suitable channels $\Lambda_i\in{\rm supp}(\nu)$:
\bb
J_{(A'B)^n}=\sum_{i=1}^Np_i(J_{A'B}^{\Lambda_i})^{\otimes n},
\ee
hence
\bb
    \EE{\Lambda\sim \nu}\Lambda^{\otimes n}=\sum_{i=1}^Np_i
    \Lambda_i^{\otimes n}.
\ee
Then,
\bb
    \mathbb{D}\left(\EE{\Lambda\sim \nu}\Lambda^{\otimes n}\,\middle\|\,\Gamma^{(n)}\right)
    &=\sup_{\rho_{RA^n}}\mathbb{D}\left(\sum_{i=1}^Np_i\Lambda_i^{\otimes n}(\rho)\,\middle\|\,\Gamma^{(n)}(\rho)\right)\\
    &\geq \sup_{\rho_{RA^n}} \min_{1\leq i\leq N} \mathbb{D}\left(\Lambda_i^{\otimes n}(\rho)\,\middle\|\,\Gamma^{(n)}(\rho)\right)-\log {\rm poly}_{d}\Big(n,\,\big|\mathrm{spec}\big(\Gamma^{(n)}(\rho)\big)\big|\Big) \\
    &\geq \sup_{\rho_{RA^n}} \min_{\Lambda\in\supp(\nu)} \mathbb{D}\left(\Lambda^{\otimes n}(\rho)\,\middle\|\,\Gamma^{(n)}(\rho)\right)-\log {\rm poly}_{d}\Big(n,\,\big|\mathrm{spec}\big(\Gamma^{(n)}(\rho)\big)\big|\Big),
\ee
where in the first inequality we have used the weak quasi-concavity of $\mathbb{D}$ and the fact that $N={\rm poly}_{d}(n)$. This concludes the proof.
\end{proof}

\begin{lemma}\label{lem:extension}
Let $\pazocal{M}_{A\to B}$ be a quantum channel, and let $\pazocal{V}^{\pazocal{M}}_{A\to BE}$ be one of its Stinespring dilations. Let $\widebar{\pazocal{M}}_{A\to BF}$ be a quantum channel that is an extension of $\pazocal{M}_{A\to B}$, in the sense that $\Tr_F \circ \widebar{\pazocal{M}}_{A\to BF} = \pazocal{M}_{A\to B}$. If $\dim \HH_F \geq \dim \HH_E$, then there exists a quantum channel $\Lambda_{E\to F}$ such that
\bb
\widebar{\pazocal{M}}_{A\to BF} = \Lambda_{E\to F}\circ \pazocal{V}^{\pazocal{M}}_{A\to BE} .
\ee
That is, all extensions of a quantum channel (up to enlarging the dimension of the extending system) can be obtained from a fixed Stinespring dilation by post-processing its environment.
\end{lemma}

\begin{proof}
Let $V^{\pazocal{M}}_{A\to BE}:\HH_A\to \HH_{BE}$ be the isometry such that $\pazocal{V}^{\pazocal{M}}_{A\to BE}(\,\cdot\,) = V^{\pazocal{M}}_{A\to BE} (\,\cdot\,) \big(V^{\pazocal{M}}_{A\to BE}\big)^\dag$, and let $W^{\widebar{\pazocal{M}}}_{A\to BFG}:\HH_A\to \HH_{BFG}$ be the isometry corresponding to a Stinespring dilation of $\widebar{\pazocal{M}}_{A\to BF}$, i.e.\ such that
\bb
\widebar{\pazocal{M}}_{A\to BF}(\,\cdot\,) = \Tr_G \Big[ W^{\widebar{\pazocal{M}}}_{A\to BFG} (\,\cdot\,) \big(W^{\widebar{\pazocal{M}}}_{A\to BFG}\big)^\dag \Big] .
\ee
Since
\bb
\dim \HH_{FG} = (\dim \HH_F) (\dim \HH_G) \geq \dim \HH_F \geq \dim \HH_E ,
\ee
elementary linear algebra considerations ensure that we can construct an isometry $Z_{E\to FG}:\HH_E\to \HH_{FG}$ with the property that
\bb
Z_{E\to FG} \circ V^{\pazocal{M}}_{A\to BE} = W^{\widebar{\pazocal{M}}}_{A\to BFG} .
\ee
Defining $\Lambda_{E\to F}(\,\cdot\,) \coloneqq \Tr_G \big[ Z_{E\to FG}^{\vphantom{\dag}} (\,\cdot\,) Z_{E\to FG}^\dag\big]$, we see that
\bb
\widebar{\pazocal{M}}_{A\to BF}(\,\cdot\,) &= \Tr_G \Big[ W^{\widebar{\pazocal{M}}}_{A\to BFG} (\,\cdot\,) \big(W^{\widebar{\pazocal{M}}}_{A\to BFG}\big)^\dag \Big] \\
&= \Tr_G \Big[ \big(Z_{E\to FG} \circ V^{\pazocal{M}}_{A\to BE}\big) (\,\cdot\,) \big(Z_{E\to FG} \circ V^{\pazocal{M}}_{A\to BE}\big)^\dag \Big] \\
&= \big(\Lambda_{E\to F}\circ \pazocal{V}^{\pazocal{M}}_{A\to BE}\big)(\,\cdot\,) ,
\ee
which concludes the proof.
\end{proof}

An analogous reasoning can be used to show the following.

\begin{lemma}\label{lemma:bound_spec}
Let $\psi_{A^n R}$ be a pure state, and let $\nu$ be a probability measure over the set of isometries from $\HH_A$ to $\HH_B$. Set $\sigma_{B^n R}\coloneqq \EE{\pazocal{V}\sim \nu}\Big[\pazocal{V}^{\otimes n}_{A\to B}(\psi_{A^n R})\Big]$. Then $\big|\mathrm{spec}(\sigma)\big|=O\big(\mathrm{poly}_{d_B}(n)\big)$, where $d_B=\mathrm{dim}\, \mathcal{H}_B$.
\end{lemma}

\begin{proof}
As in the proof of Lemma~\ref{lem:weak_qc}, applying Carath\'{e}odory's theorem to the Choi state of the channel $\EE{\pazocal{V}\sim \nu}\pazocal{V}^{\otimes n}_{A\to B}$, which belongs to the real vector space of Hermitian permutationally symmetric operators on $\HH_{A'B}^{\otimes n}$, we can write 
\bb
\EE{\pazocal{V}\sim \nu} \pazocal{V}^{\otimes n}_{A\to B} = \sum_{i=1}^N p_i \pazocal{V}^{\otimes n}_{i}\, ,
\label{lemma:bound_spec_proof_eq1}
\ee
for some choice of isometries $\pazocal{V}_{i}:\HH_A\to \HH_B$, $i=1,\ldots,N$, and 
\bb
N=(n+1)^{(d_A d_B)^2-1} + 1 \leq (n+1)^{d_B^4-1} + 1\, .
\ee
Applying~\eqref{lemma:bound_spec_proof_eq1} to $\psi_{A^nR}$ and noticing that $\big|\mathrm{spec}\big(\sumno_{i=1}^N p_i \Psi_i \big)\big| \leq N+1$ directly shows the claim.
%
\end{proof}

Now we have all the ingredients to state and prove a completely new characterisation of the relative entropy between channels in terms of their extensions.

\begin{boxedthm}{}
\begin{thm}[(Axiomatic Uhlmann's theorem for channels)]\label{thm:Uhlmann2}
Let $\mathbb{D}(\,\cdot\,\|\,\cdot\,)$ be a jointly convex divergence that obeys weak quasi-concavity, i.e.~\eqref{eq_weak_quasi_conc}. 
Let $\pazocal{M}_{A\to B}$ and $\pazocal{N}_{A\to B}$ be quantum channels from $\mathcal{H}_A$ to $\mathcal{H}_B$, and let $\mathcal{H}_E$ be a Hilbert space of dimension $\dim \mathcal{H}_A\cdot\dim\mathcal{H}_B$. Then, for any arbitrary $E$-dilation $\widebar{\pazocal{M}}_{A\to BE}$ of $\pazocal{M}_{A\to B}$, we have
\bb\label{eq:Uhlmann2}
    \mathbb{D}^{\infty}(\pazocal{M}\|\pazocal{N})= \Rel{\mathbb{D}^\infty}{\widebar{\pazocal{M}}_{A\to BE}}{\mathcal{C}_{A\to BE}^{\pazocal{N}}}\, .
\ee
Moreover, a sequence of asymptotic optimisers $\big(\widebar{\pazocal{N}}_{A^n\to B^nE^n}\big)_n\in \mathcal{C}^{\pazocal{N}}_{A\to BE}$ is
\bb\label{eq:almost_optimizers2}
    \widebar{\pazocal{N}}_{A^n\to B^nE^n}=\Lambda_{E\to E}^{\otimes n}\circ\pazocal{D}_{B^nM\to B^nE^n}\circ \pazocal{N}^{\otimes n}_{A\to B}\circ \pazocal{E}_{A^n\to A^nM}\,,
\ee
where $\Lambda_{E\to E}$ is any channel that, by acting only on the auxiliary system $E$, maps a fixed $E$-dilation of $\pazocal{M}_{A\to B}$ to the chosen $E$-dilation $\widebar{\pazocal{M}}_{A\to BE}$, and, for each $n\geq 1$, $\pazocal{E}$ and $\pazocal{D}$ are the encoder and the decoder channels defined in Theorem~\ref{thm1}, respectively.
\end{thm}
\end{boxedthm}

\begin{proof}
The inequality $\mathbb{D}^{\infty}(\pazocal{M}\|\pazocal{N})\leq  \mathbb{D}^\infty\big(\widebar{\pazocal{M}}_{A\to BE}\,\big\|\,\mathcal{C}_{A\to BE}^{\pazocal{N}}\big)$ immediately follows from the data-processing inequality for $\mathbb{D}$, by applying the channel ${\rm Id}_{B^n}\otimes\Tr_{E^n}[\,\cdot\,]$ in the very definition of the right-hand-side of~\eqref{eq:Uhlmann2} for any $n\geq 1$. 

Let us now prove the converse inequality. First, it suffices to consider the case where the $E$-dilation 
\bb 
\widebar{\pazocal{M}}_{A\to BE}(\,\cdot\,) = \pazocal{V}^{\pazocal{M}}_{A\to BE}(\,\cdot\,)\coloneqq V^{\pazocal{M}}_{A\to BE}\,\cdot\,V^{\pazocal{M}\; \dagger}_{A\to BE}
\ee
of $\pazocal{M}_{A\to B}$ is an isometry. Indeed, by Lemma~\ref{lem:extension}, any other $E$-dilation $\widebar{\pazocal{M}}'_{A\to BE}$ can be obtained by applying a suitable quantum channel $\Lambda_{E\to E}$ to the auxiliary system: 
\bb
\widebar{\pazocal{M}}'_{A\to BE} =({\rm Id}_B\otimes \Lambda_{E\to E})\circ \pazocal{V}^{\pazocal{M}}_{A\to BE}.
\ee
Hence, 
\bb
\inf_{\widebar{\pazocal{N}}\in \pazocal{C}_{n}^{\pazocal{N}}} \Rel{\mathbb{D}}{\widebar{\pazocal{M}}'^{\,\otimes n}_{A\to BE}}{\widebar{\pazocal{N}}_{A^n\to B^nE^n}} &\leq \inf_{\widetilde{\pazocal{N}}\in \pazocal{C}_{n}^{\pazocal{N}}} \Rel{\mathbb{D}}{\Lambda_{E\to E'}^{\otimes n}\circ \pazocal{V}^{\pazocal{M}\;\otimes n}_{A\to BE}}{\Lambda_{E\to E'}^{\otimes n}\circ \widetilde{\pazocal{N}}_{A^n\to B^nE^n}} \\
&\leq \inf_{\widetilde{\pazocal{N}}\in \pazocal{C}_{n}^{\pazocal{N}}} \Rel{\mathbb{D}}{\pazocal{V}^{\pazocal{M}\;\otimes n}_{A\to BE}}{\widetilde{\pazocal{N}}_{A^n\to B^nE^n}}\, .
\ee
Here, the first inequality holds by taking as ansatzes all $E$-dilations of $\pazocal{N}_{A\to B}$ of the form $\Lambda_{E\to E'}^{\otimes n}\circ \widetilde{\pazocal{N}}_{A^n\to B^nE^n}$, where $\widetilde{\pazocal{N}}_{A^n\to B^nE^n} \in \pazocal{C}^{\pazocal{N}}_n \coloneqq \pazocal{C}^{\pazocal{N}^{\otimes n}}_{A^n\to B^nE^n}$; the second inequality, instead, is simply data-processing.
Now we are going to show that the right-hand-side of the above equation is upper bounded by $n\mathbb{D}^\infty(\pazocal{M}\|\pazocal{N})$ up to terms that are sublinear in $n$; this will complete the proof. To this end, we lower bound
\begin{align}\label{eq:inequalities}
    \frac{1}{n}\mathbb{D}\left(\pazocal{M}_{A\to B}^{\otimes n}\middle\|\pazocal{N}_{A\to B}^{\otimes n}\right)
    &= \sup_{\rho_{A^nR}}\frac{1}{n}\mathbb{D}\left(\pazocal{M}_{A\to B}^{\otimes n}(\rho_{A^nR})\middle\|\pazocal{N}_{A\to B}^{\otimes n}(\rho_{A^nR})\right)\\
    \nonumber& \geqt{(i)} \sup_{\tilde \rho_{A^nR'}}\frac{1}{n}\mathbb{D}\left((\pazocal{M}_{A\to B}^{\otimes n}\circ \pazocal{E}_{A^n\to A^nM})(\tilde\rho_{A^nR'}) 
    \middle\|(\pazocal{N}_{A\to B}^{\otimes n}\circ \pazocal{E}_{A^n\to A^nM})(\tilde\rho_{A^nR'})\right)\\
    \nonumber&\geqt{(ii)}\sup_{\tilde \rho_{A^nR'}}\frac{1}{n}\mathbb{D}\left(\EE{\pazocal{V}^{\pazocal{M}}}\Big[\pazocal{V}^{\pazocal{M}\;\otimes n}_{A\to BE}(\tilde \rho_{A^nR'})\Big]\,\middle\|\, \EE{\pazocal{V}^{\pazocal{N}}}\Big[\pazocal{V}^{\pazocal{N}\;\otimes n}_{A\to BE}(\tilde \rho_{A^nR'})\Big]\right)\\[4pt]
    \nonumber&\geqt{(iii)}\sup_{\tilde \rho_{A^nR'}}\frac{1}{n}\min_{\pazocal{V}^{\pazocal{M}}}\mathbb{D}\left(\pazocal{V}^{\pazocal{M}\;\otimes n}_{A\to BE}(\tilde \rho_{A^nR'})\,\middle\|\, \EE{\pazocal{V}^{\pazocal{N}}}\Big[\pazocal{V}^{\pazocal{N}\;\otimes n}_{A\to BE}(\tilde\rho_{A^nR'})\Big]\right)-\tfrac{\log{\rm poly }(n)}{n}\\[4pt]
    \nonumber& \eqt{(iv)}\frac{1}{n}\mathbb{D}\left(\widebar {\pazocal{V}}^{\pazocal{M}\;\otimes n}_{A\to BE}\,\middle\|\, \EE{\pazocal{V}^{\pazocal{N}}}\Big[\pazocal{V}^{\pazocal{N}\;\otimes n}_{A\to BE}\Big]\right)-\tfrac{\log{\rm poly }(n)}{n}\\[4pt]
    \nonumber&\geqt{(v)}\frac{1}{n}\inf_{\tilde{\pazocal{N}}\in\pazocal{C}_n^{\pazocal{N}}}\mathbb{D}\left(\widebar {\pazocal{V}}^{\pazocal{M}\;\otimes n}_{A\to BE}\,\middle\|\, \tilde{\pazocal{N}}\right)-\tfrac{\log{\rm poly }(n)}{n},
\end{align}
where $\pazocal{C}^{\pazocal{N}}_n \coloneqq \pazocal{C}^{\pazocal{N}^{\otimes n}}_{A^n\to B^nE^n}$, as before, and 
\bb
\EE{\pazocal{V}^{\pazocal{M}}}\Big[\pazocal{V}^{\pazocal{M}\;\otimes n}_{A\to BE}(\tilde \rho_{A^nR'})\Big]\coloneqq \EE{U_E}\!\left[
\big((\id_B\otimes U_E)V^{\pazocal{M}}_{A\to BE}\big)^{\otimes n}
(\tilde \rho_{A^nR'})
\big(V_{A\to BE}^{\pazocal{M}\;\dagger}(\id_B\otimes U_E^\dagger)\big)^{\otimes n}
\right];
\ee
in particular,
\begin{itemize}
    \item in (i) we have chosen the auxiliary system $R$ to  be of the form $R=MR'$, with $R'$ arbitrary and $M$ being the memory system appearing in Theorem~\ref{thm1}, and we have restricted the supremum to states $\rho_{A^nR}$ of the form $(\pazocal{E}_{A^n\to A^n M}\otimes {\rm Id}_{R'})(\tilde\rho_{A^nR'})$, where $\tilde\rho_{A^nR'}$ is arbitrary and $\pazocal{E}_{A^n\to A^nM}$ is the encoder introduced in Theorem~\ref{thm1};
    \item the lower bound in (ii) is the data-processing inequality when applying the decoding channel $\pazocal{D}_{B^nM\to B^nE^n}$ of Theorem~\ref{thm1} to both arguments of the divergence $\mathbb{D}$; as a result, by~\eqref{eq_random_isometry}, we get $n$ copies of the random Stinespring dilations $\pazocal{V}^{\pazocal{M}}_{A\to BE}$ of $\pazocal{M}$ and $\pazocal{V}^{\pazocal{N}}_{A\to BE}$ of $\pazocal{N}$, respectively;
    \item in (iii) we have leveraged Lemmas~\ref{lem:weak_qc} and~\ref{lemma:bound_spec}, noting that the supremum can be restricted to pure states due to joint convexity of $\mathbb{D}$;
    \item in (iv) we have noticed that the function of $\pazocal{V}^{\pazocal{M}}_{A\to BE}$ to be minimised actually is independent of $\pazocal{V}^{\pazocal{M}}_{A\to BE}$, therefore we can choose any arbitrary fixed dilation $\widebar{\pazocal{V}}^{\pazocal{M}}_{A\to BE}$; indeed, for any fixed Stinespring dilation $\pazocal{V}^{\pazocal{M}}_{A\to BE}$, we can apply a local unitary channel $\pazocal{U}_E$ on the system $E$ to get $\widebar{\pazocal{V}}^{\pazocal{M}}_{A\to BE}$; in particular, by the unitary invariance of $\mathbb{D}$ --- which follows from the data-processing inequality --- and by the left-invariance of the Haar measure, we have
\bb
    &\mathbb{D}\left(\pazocal{V}^{\pazocal{M}\;\otimes n}_{A\to BE}(\rho_{A^nR'})\middle\| \;\EE{\pazocal{V}^{\pazocal{N}}}\Big[\pazocal{V}^{\pazocal{N}\;\otimes n}_{A\to BE}(\rho_{A^nR'})\Big]\right)\\
    &\qquad =\mathbb{D}\left(\big(\pazocal{U}_{E}\circ\pazocal{V}^{\pazocal{M}}_{A\to BE}\big)^{\otimes n}(\rho_{A^nR'})\middle\| \,\EE{\pazocal{V}^{\pazocal{N}}}\Big[\big(\pazocal{U}_E\circ\pazocal{V}^{\pazocal{N}}_{A\to BE}\big)^{\otimes n}(\rho_{A^nR'})\Big]\right)\\
    &\qquad=\mathbb{D}\left(\widebar{\pazocal{V}}^{\pazocal{M}\;\otimes n}_{A\to BE}(\rho_{A^nR'})\middle\| \;\EE{\pazocal{V}^{\pazocal{N}}}\Big[\pazocal{V}^{\pazocal{N}\;\otimes n}_{A\to BE}(\rho_{A^nR'})\Big]\right);
\ee
    \item finally, in (v) we have noticed that $\EE{\pazocal{V}^{\pazocal{N}}}\Big[\pazocal{V}^{\pazocal{N}\;\otimes n}_{A\to BE}\Big]\in \pazocal{C}_{A^n\to B^n E^n}^{\pazocal{N}^{\otimes n}}$.
\end{itemize}
Taking the limit $n\to\infty$ in~\eqref{eq:inequalities}, we get
\bb
    \mathbb{D}^{\infty}(\pazocal{M}\|\pazocal{N}) \geq \liminf_{n\rightarrow\infty}\frac{1}{n}\inf_{\tilde{\pazocal{N}}\in\pazocal{C}_n^{\pazocal{N}}}\mathbb{D}\left(\widebar {\pazocal{V}}^{\pazocal{M}\;\otimes n}_{A\to BE}\,\middle\|\, \tilde{\pazocal{N}}\right)
    =\mathbb{D}^\infty\big(\widebar{\pazocal{M}}_{A\to BE}\,\big\|\,\mathcal{C}_{A\to BE}^{\pazocal{N}}\big).
\ee
In particular, this proof implies that the sequence $\big(\widebar{\pazocal{N}}_{A^n\to B^nE^n}\big)_n\in \mathcal{C}^{\pazocal{N}}_{A\to BE}$ given by~\eqref{eq:almost_optimizers2} achieves the right-hand side of~\eqref{eq:Uhlmann2}.
\end{proof}

\section{Applications to quantum learning theory} \label{appl_learning}


In this section, we apply the random Stinespring superchannel to quantum learning theory, focusing on the problem of quantum channel learning~\cite{AMele2025, chen2025quantumchanneltomographyestimation}. More specifically, our construction reduces tomography of general quantum channels to tomography of isometries, leading to the optimal query complexity \(O(r d_A d_B)\) for learning rank-\(r\) quantum channels, recently established in~\cite{AMele2025,chen2025quantumchanneltomographyestimation}. We also develop an algebraic lower bound technique based on polynomial method~\cite{beals2001quantum} that allows us to prove a clean $\Omega(rd_Ad_B)$ lower bound without any logarithmic factors and is secure against arbitrary types of queries (e.g.\ queries to the inverse or controlled versions of the channel, or with indefinite causal order).
Together, this establishes $\Theta(rd_Ad_B)$ as the optimal query complexity of learning rank-$r$ quantum channels.

We note that alternative learning algorithms achieving the same $O(rd_Ad_B/\epsilon^2)$ query complexity have recently been developed in~\cite{AMele2025} via a random purification channel on Choi states and in~\cite{chen2025quantumchanneltomographyestimation} via a tester-dependent random dilation procedure (i.e.\ its construction depends explicitly on both the input state and on the measurement carried out at the output). Our results provide an alternative, state- and measurement-agnostic way to reduce channel learning to isometry learning. Meanwhile, the only known lower bound for general non-isometry channels is $\Omega(d_A^2d_B^2/\log(d_Ad_B))$ when the channels have full rank ($r=d_Ad_B$).
It has undesired logarithmic factors and holds when we only allow sequential queries of the channel~\cite{rosenthal2024quantum}. We now restate the main result of this section in the form of a quotable theorem.
\begin{boxedthm}{}
\begin{thm}[(Optimal query complexity of channel learning)]
\label{thm:channel-learning}
    Let $\Phi: \mathcal{L}(\mathcal{H}_A) \to \mathcal{L}(\mathcal{H}_B)$ be any quantum channel with input dimension $d_A$, output dimension $d_B\geq 2$, and rank at most $r$.
    From~\cite{AMele2025,chen2025quantumchanneltomographyestimation} it is known that there is a quantum learning algorithm that makes
    \begin{equation}
        n = O (rd_Ad_B/\epsilon^2)
    \end{equation}
    parallel queries to the channel $\Phi$ and outputs a classical description of a channel $\hat{\Phi}$ such that
    $\|\hat{\Phi}-\Phi\|_\diamond\leq \epsilon$ with probability at least $2/3$.
    Furthermore, any quantum algorithm that learns $\Phi$ to constant error with success probability at least $2/3$ must make at least
    \begin{equation}
        n\geq \Omega(rd_Ad_B)
    \end{equation}
    queries even if it is allowed to query $\Phi$ in an arbitrary way (e.g.\ query the inverse and controlled versions of $\Phi$ if they exist, or with indefinite causal order).
\end{thm}
\end{boxedthm}

\noindent \textbf{Note added.} The improved lower bound on the query complexity of channel learning presented in Theorem~\ref{thm:channel-learning}, without logarithmic factors, has been obtained independently in version~2 of~\cite{AMele2025}.
\medskip

The upper bound in Theorem~\ref{thm:channel-learning} follows immediately from our random Stinespring superchannel and existing isometry learning algorithms, such as the one provided in~\cite{AMele2025} via Choi-state learning or the one in~\cite[Appendix A]{chen2025quantumchanneltomographyestimation}, which is a slight modification of the unitary tomography algorithm in~\cite{haah2023query}. These subroutines of isometry learning only make parallel queries to the isometry and therefore our random Stinespring superchannel can be directly applied.

On the lower bound front, it is intuitive that an $\Omega(rd_Ad_B)$ bound should hold by dimension counting.
But the only lower bound known for general non-isometry channels is $\Omega(d_A^2d_B^2/\log(d_Ad_B))$ when $r=d_Ad_B$ with an undesired logarithmic factor and holds when we only allow sequential queries of the channel~\cite{rosenthal2024quantum}.
This logarithmic factor comes from a crude information-theoretic analysis that does not take into account the permutation symmetry between queries of the channel (i.e.\ they are the same channel).
A natural way to make use of the permutation symmetry is via the heavy machinery of group representation theory.
For example,~\cite{haah2023query} shows an $\Omega(d^2/\epsilon)$ lower bound for unitary tomography ($r=1, d_A=d_B=d$) that does not have any logarithmic factor using a unitary distinguishing bound~\cite{bavaresco2022unitary} proved via Schur--Weyl duality.

However, this route is undesirable for several reasons: (1) it uses heavy group representation theory machinery that departs significantly from our simple intuition of dimension counting; (2) whether it can be generalised to channels is unclear since channels do not even form a group; (3) when we are allowed to make queries to the inverse or controlled versions of the channel (if they exist), the queries are no longer permutation symmetric.

To overcome these difficulties, we develop a lower bound proof technique that is purely algebraic and only depends on the linearity of quantum mechanics.
It completely circumvents the representation theory machinery and reduces everything to simple dimension counting.
The permutation symmetry is then used transparently in dimension counting.
This allows us to prove a channel distinguishing bound that extends the unitary version~\cite{bavaresco2022unitary} and is secure against any type of queries to the channel.
The claimed query lower bound for channel tomography follows directly from the standard reduction from learning to distinguishing using packing net.
We expect that this proof strategy can be applied to quantum channels with other parameterisations beyond bounded rank.

In the following, we detail the proof of the lower bound in Theorem~\ref{thm:channel-learning}.
We begin by explaining the algebraic proof that leads to the following channel distinguishing bound.

\begin{thm}[(Polynomial method for channel distinguishing)]
\label{lem:limitation}
    Let $\{\Phi_x\}_{x=1}^M$ be a set of quantum channels with input dimension $d_A$, output dimension $d_B$, and Choi rank $r\leq d_Ad_B$. 
    Given any quantum channel $\Phi_x, x\in [M]$ from the set, any quantum algorithm (even with indefinite causal order) that makes $n$ queries to the channel $\Phi_x$ and produces an outcome $\hat{X}\in [M]$ with correct probability $\Pr[\hat{X}=x|\Phi_x]>1/2$ for any $x\in [M]$ must satisfy
    \begin{equation}
        \frac{1}{2}\log M\leq \log \binom{n+rd_Ad_B-1}{n}.
    \end{equation}
    This still holds when the quantum algorithm is allowed to query the inverse and controlled versions of $\Phi_x$ if they exist.
\end{thm}

\begin{proof}
    The proof generalises the polynomial method developed in~\cite{beals2001quantum,huang2021information}.
    The key idea is to exploit the linearity of quantum mechanics, which implies that the measurement probability of any quantum algorithm that makes $n$ queries to a channel must be a polynomial of the channel parameters with degree determined by $n$.
    But the degree cannot be too small in order to distinguish many channels.
    This gives a lower bound on the query complexity $n$.
    
    Suppose that there is a quantum algorithm that makes $n$ queries to the channel $\Phi_x$ and produces an outcome $\hat{X}\in [M]$ that has correct probability $\Pr[\hat{X}=x|\Phi_x]>1/2$ for any $x\in [M]$.
    We consider the confusion matrix $P \in \mathbb{R}^{M\times M}$ of this quantum algorithm.
    It is an $M\times M$ matrix with matrix elements $P_{\hat{x}x}, \hat{x}, x\in [M]$ representing the probability that the quantum algorithm predicts $\hat{x}$ when the quantum channel that it truly queries is $\Phi_x$.
    We have $\sum_{\hat{x}=1}^M P_{\hat{x}x}=1$ for all $x\in [M]$.
    The guarantee of correct probability implies that $P_{xx}>1/2$ for any $x\in [M]$.
    Therefore, 
    \begin{equation}
        \sum_{\hat{x}\neq x}P_{\hat{x}x}=1-P_{xx}< \frac{1}{2}< P_{xx}, \quad \forall x\in [M],
    \end{equation}
    meaning that the confusion matrix $P$ is strictly diagonally dominant.
    This implies that $P$ has full rank:
    \begin{equation}
        \mathrm{rank}(P)=M.
    \end{equation}
    On the other hand, all these matrix elements are measurement probabilities of a quantum algorithm querying the channel $\Phi_x$.
    Let
    \begin{equation}
        \Phi_x(\rho)=\sum_{i=1}^{r} K^x_i(\rho)K^{x\dagger}_i
    \end{equation}
    be the Kraus operator representation of the quantum channel $\Phi_x$.
    We use a complex vector $z^x\in \mathbb{C}^{rd_Ad_B}$ to collect all the parameters in the Kraus operators:
    \begin{equation}
        z^x_{i\cdot d_Ad_B+j\cdot d_A+k} = (K^x_i)_{jk}.
    \end{equation}
    Then the matrix elements of the output of the quantum channel $\Phi_x$ can be regarded as a polynomial of $z^x$ and its complex conjugate $\bar{z}^x$:
    \begin{equation}
        (\Phi_x(\rho))_{ij} = \sum_{a,b=1}^{rd_Ad_B} w_{ij,ab}z^x_a \bar{z}^x_b, \quad \forall i,j \in [d_B],
    \end{equation}
    where the coefficients $w_{ij,ab}\in \mathbb{C}$ are determined by the input state $\rho$.

    We generalise this polynomial representation to the measurement probability of an arbitrary quantum algorithm (possibly with indefinite causal order) querying the channel $\Phi_x$.
    The most general form of the measurement probability $P_{\hat{x}x}$ is represented as the contraction of a general \emph{algorithm tensor} $(T_{\hat{x}})_{i_1i'_1o_1o'_1\ldots i_ni'_no_no'_n}, i_1, \ldots, i_n, i'_1, \ldots, i'_n\in [d_A], o_1, \ldots, o_n, o'_1, \ldots, o'_n\in [d_B]$ with $n$ copies of the channel tensor $(\Phi_x)_{ii'oo'}, i,i'\in [d_A], o,o'\in [d_B]$:
    \begin{equation}
        P_{\hat{x}x} = \sum_{\substack{i_1, \ldots, i_n, i'_1, \ldots, i'_n\in [d_A]\\ o_1, \ldots, o_n, o'_1, \ldots, o'_n\in [d_B]}} (T_{\hat{x}})_{i_1i'_1o_1o'_1\ldots i_ni'_no_no'_n} (\Phi_x)_{i_1i'_1o_1o'_1}\cdots (\Phi_x)_{i_ni'_no_no'_n}.
    \end{equation}
    Here, for each $j\in [n]$, the indices $i_ji'_jo_jo'_j$ of the tensor $T_{\hat{x}}$ are contracted with the $n$-th copy of the channel $\Phi_x$.
    Note that the tensor $T_{\hat{x}}$ must satisfy certain conditions to ensure that the outcome is a proper probability (e.g.\ $P_{\hat{x}x}\in [0, 1], \sum_{\hat{x}}P_{\hat{x}x}=1$), but for our purposes we do not use those conditions.
    Plugging in the $z^x, \bar{z}^x$ parameterisation of the channel $\Phi_x$, we have
    \begin{equation}
        P_{\hat{x}x} = \sum_{\substack{a_1, \ldots, a_n\in [rd_A d_B]\\ b_1, \ldots, b_n\in [rd_A d_B]}} w^{\hat{x}}_{a_1\ldots a_n b_1\ldots b_n} z^x_{a_1}\cdots z^x_{a_n} \bar{z}^x_{b_1}\cdots \bar{z}^x_{b_n},
    \end{equation}
    where the coefficients $w^{\hat{x}}_{a_1\ldots a_n b_1\ldots b_n}$ are determined by the algorithm tensor $T_{\hat{x}}$ and the terms $z^x_{a_j} \bar{z}^x_{b_j}$ that are contributed by the $j$-th copy of the channel tensor $\Phi_x$.
    
    We note that this way of organising coefficients has redundancy, because the $n$ copies of $z$'s (and $\bar{z}$'s) are symmetric to each other.
    For example, the terms $z_1z_2$ and $z_2z_1$ are the same and can be grouped together to share one coefficient.
    In other words, the order in the indices $(a_1, \ldots, a_n)$ and $(b_1, \ldots, b_n)$ does not matter.
    We use $\mathrm{Sym}(rd_Ad_B, n)$ to denote the set of such unordered indices and use $Z^x_\alpha, \bar{Z}^x_\beta \in \mathbb{C}, \alpha, \beta\in \mathrm{Sym}(rd_Ad_B, n)$ to denote the terms $z^x_{a_1}\cdots z^x_{a_n}, \bar{z}^x_{b_1}\cdots \bar{z}^x_{b_n}$ corresponding to the unordered indices $\alpha, \beta$.
    Then we have the following polynomial representation of the probability
    \begin{equation}
        P_{\hat{x}x} = \sum_{\alpha, \beta\in \mathrm{Sym}(rd_Ad_B, n)} W^{\hat{x}}_{\alpha\beta} Z^x_{\alpha} \bar{Z}^x_\beta,
    \end{equation}
    where the coefficients $W^{\hat{x}}_{\alpha\beta}$ are the sum of all $w^{\hat{x}}_{a_1\ldots a_n b_1\ldots b_n}$ with $(a_1, \ldots, a_n), (b_1, \ldots, b_n)$ corresponding to $\alpha, \beta$.
    To count the size of $\mathrm{Sym}(rd_Ad_B, n)$, we note that each $\alpha \in \mathrm{Sym}(rd_Ad_B, n)$ can be labeled by the number of times $n_a$ each symbol $a\in [rd_Ad_B]$ appears in the unordered indices $\alpha$.
    They satisfy
    \begin{equation}
        \sum_{a=1}^{rd_Ad_B}n_a=n, \quad n_a\geq 0, \quad \forall a\in [rd_Ad_B].
    \end{equation}
    Standard combinatorial counting yields
    \begin{equation}
        |\mathrm{Sym}(rd_Ad_B, n)| = \binom{n+rd_Ad_B-1}{n}.
    \end{equation}
    The polynomial representation gives us a matrix decomposition of the confusion matrix $P$:
    \begin{equation}
        P = \mathcal{W}\mathcal{Z},
    \end{equation}
    where the matrices
    \begin{equation}
        \mathcal{W}\in \mathbb{C}^{M\times |\mathrm{Sym}(rd_Ad_B, n)|^2}, \quad \mathcal{Z}\in \mathbb{C}^{|\mathrm{Sym}(rd_Ad_B, n)|^2\times M},
    \end{equation}
    are given by the coefficients $W^{\hat{x}}_{\alpha\beta}$ and monomials $Z^x_{\alpha}\bar{Z}^x_{\beta}$: for each $x, \hat{x}\in [M]$, the $\hat{x}$-th row of $\mathcal{W}$ is the row vector $(W^{\hat{x}}_{\alpha\beta})_{\alpha, \beta \in \mathrm{Sym}(rd_Ad_B, n)}$ and the $x$-th column of $\mathcal{Z}$ is the column vector $((Z^x_{\alpha}\bar{Z}^x_{\beta})_{\alpha, \beta \in \mathrm{Sym}(rd_Ad_B, n)})^T$.
    In other words,
    \begin{equation}
        P_{\hat{x}x} = \sum_{\gamma=1}^{|\mathrm{Sym}(rd_Ad_B, n)|^2}\mathcal{W}_{\hat{x},\gamma}\mathcal{Z}_{\gamma, x}.
    \end{equation}
    Therefore, the rank of the confusion matrix satisfies
    \begin{equation}
        M=\mathrm{rank}(P)\leq |\mathrm{Sym}(rd_Ad_B, n)|^2 = \binom{n+rd_Ad_B-1}{n}^2.
    \end{equation}
    Taking the logarithm, we arrive at the desired result
    \begin{equation}
        \log M \leq 2 \log \binom{n+rd_Ad_B-1}{n}.
    \end{equation}

    When we are allowed to query the inverse and controlled versions of $\Phi_x$ if they exist, the contracted channel tensor is the same as that of $\Phi_x$ itself with $z^x$ and $\bar{z}^x$ swapped or padded with fixed numbers that represent the control pattern.
    This does not change the polynomial representation and the counting.
    Therefore, we still have
    \begin{equation}
        \log M \leq 2 \log \binom{n+rd_Ad_B-1}{n}.
    \end{equation}
    This completes the proof of Theorem~\ref{lem:limitation}.
\end{proof}

To prove a query complexity lower bound for learning, we instantiate the $M$ quantum channels with the maximal cardinality while keeping their distinguishability under a learning algorithm.
This can be done by constructing an $\epsilon$-packing net of the set of rank-$r$ channels.

\begin{defi}[(Packing net)]
    Let $(X,d)$ be a metric space. Let $K \subseteq X$ be a subset and $\epsilon > 0$. 
    Then, a subset $N \subseteq K$ is an \emph{$\epsilon$-packing net} of $K$ if for any $x,y \in N$, $d(x,y) > \epsilon$. 
    The \emph{packing number} $\mathcal{M}(K, d, \epsilon)$ of $K$ is the largest possible cardinality of an $\epsilon$-packing net of $K$.
\end{defi}

To construct a packing net for  channels, we first construct a packing net for isometries.

\begin{lemma}[(Packing number of isometries~\cite{szarek1997metric})]
\label{lem:packing-iso}
    Let $d_2\geq d_1$ be positive integers and $\|\cdot \|$ be the operator norm.
    Let $\mathcal{V}_{d_1\to d_2} = \{V\in \mathbb{C}^{d_2\times d_1}: V^\dagger V = \id_{d_1}\}$ be the set of isometries with input dimension $d_1$ and output dimension $d_2$, also known as the Stiefel manifold.
    It has dimension $\dim(\mathcal{V}_{d_1\to d_2}) = 2d_1d_2 - d_1^2 \in [d_1d_2, 2d_1d_2]$ and packing number
    \begin{equation}
        \left( \frac{C_1}{\epsilon} \right)^{\dim(\mathcal{V}_{d_1\to d_2})}\leq \mathcal{M}(\mathcal{V}_{d_1\to d_2}, \|\cdot \|, \epsilon) \leq \left( \frac{C_2}{\epsilon} \right)^{\dim(\mathcal{V}_{d_1\to d_2})}
    \end{equation}
    for some universal constants $C_1,C_2>0$.
    In particular, when $d_2=d_1=d$, we have that the packing number of the $d$-dimensional unitary group satisfies
    \begin{equation}
        \left( \frac{C_1}{\epsilon} \right)^{d^2}\leq \mathcal{M}(\mathcal{V}_{d\to d}, \|\cdot \|, \epsilon) \leq \left( \frac{C_2}{\epsilon} \right)^{d^2}.
    \end{equation}
\end{lemma}

The diamond norm distance between channels is connected with the operator norm distance of their Stinespring dilations via the following continuity lemma.

\begin{lemma}[(Continuity of Stinespring dilation~\cite{kretschmann2008information})]
\label{lem:cont-dilation}
    Let $\Phi_1, \Phi_2: \mathcal{L}(\mathcal{H}_A)\to \mathcal{L}(\mathcal{H}_B)$ be two quantum channels with Stinespring dilations $V_1, V_2: \mathcal{H}_A\to \mathcal{H}_B\otimes \mathcal{H}_E$.
    Then, we have
    \begin{equation}
	\inf_{U} \|(\id_B\otimes U)V_1-V_2\|^2 \leq \|\Phi_1-\Phi_2\|_\diamond \leq 2\inf_{U}\|(\id_B\otimes U)V_1-V_2\|,
    \end{equation}
    where the infimum is over all unitary $U$ on $\mathcal{H}_E$, $\|\cdot\|_\diamond$ is the diamond norm, and $\|\cdot\|$ is the operator norm.
\end{lemma}

This shows that channels can be viewed as isometries with the unitary group on the environment quotient out.
This observation enables us to bound the packing number of channels in diamond norm as follows.

\begin{lemma}[(Packing number of channels)]
\label{lem:packing-channel}
    Let $\mathcal{C}_{d_A, d_B, r}$ be the set of quantum channels with input dimension $d_A$, output dimension $d_B$, and rank $r$.
    Assume that $d_B\geq 2$.
    We have
    \begin{equation}
        \log\mathcal{M}(\mathcal{C}_{d_A, d_B, r}, \|\cdot \|_\diamond, \epsilon) = \Theta\left(rd_A d_B \log(1/\epsilon)\right).
    \end{equation}
\end{lemma}

\begin{proof}
    The proof of Lemma~4 in~\cite[arXiv version]{barthel2018fundamental} (see also Lemma~10 in~\cite{zhao2024learning}), combined with Lemma~\ref{lem:cont-dilation}, shows that $\mathcal{M}(\mathcal{C}_{d_A, d_B, r}, \|\cdot \|_\diamond, \epsilon)$ is asymptotically bounded from both sides by the packing number of $\mathcal{V}_{d_A\to rd_B}$ in $\|\cdot \|$ divided by the packing number of $\mathcal{V}_{r\to r}$ in $\|\cdot \|$ up to a quadratic difference in $\epsilon$.
    When we take the logarithm, the division becomes subtraction.
    Using Lemma~\ref{lem:packing-iso}, we have
    \begin{equation}
    \begin{split}
        \log\mathcal{M}(\mathcal{C}_{d_A, d_B, r}, \|\cdot \|_\diamond, \epsilon) &= \Theta((2rd_Ad_B-d_A^2)\log(1/\epsilon)) - \Theta(r^2\log(1/\epsilon))\\
        &=\Theta((2rd_Ad_B-d_A^2-r^2)\log(1/\epsilon))
    \end{split}
    \end{equation}
    Further note that $2rd_Ad_B-d_A^2 - r^2\leq 2rd_Ad_B$ and 
    \begin{equation}
    \begin{split}
        2rd_Ad_B-d_A^2 - r^2 &= rd_Ad_B \left( 2 - \frac{1}{d_B}\left(\frac{d_A}{r} +  \frac{r}{d_A}\right) \right) \\
        &\geq rd_Ad_B \left( 2 - \frac{1}{d_B}\left(d_B + \frac{1}{d_B}\right) \right) \\
        &=rd_Ad_B\left( 1-\frac{1}{d_B^2} \right) \\
        &\geq \frac{3}{4}rd_Ad_B,
    \end{split}
    \end{equation}
    when $d_B\geq 2$.
    Here, we used the fact that $r\leq d_Ad_B$ and $rd_B\geq d_A$, and that the function $f(x)=x+1/x$ is convex and hence its maximum on $d_A/r\in [1/d_B, d_B]$ must be attained at the endpoints.
    This means that $2rd_Ad_B-d_A^2 - r^2 = \Theta(rd_Ad_B)$ and therefore
    \begin{equation}
        \log\mathcal{M}(\mathcal{C}_{d_A, d_B, r}, \|\cdot \|_\diamond, \epsilon) = \Theta\left(rd_A d_B \log(1/\epsilon)\right).
    \end{equation}
\end{proof}

The following lemma helps us work through the binomial factors and calculate the query complexity bound.

\begin{lemma}
\label{lem:binomial}
    Let $n, d$ be positive integers.
    Suppose $\log \binom{n+d-1}{n}\geq c(d-1)$ for some constant $c>0$; then $n\geq g^{-1}(c)(d-1)$, where $g(x)= (1+x)\log(1+x)-x\log(x)$, called the `bosonic entropy function', is monotonically increasing.
\end{lemma}

\begin{proof}
    When $d=1$, the lemma clearly holds.
    When $d\geq 2$, we begin by relating the log binomial coefficient to the binary entropy function $H_2(p) \coloneqq -p\log p-(1-p)\log(1-p)$.
    Note that
    \begin{equation}
    \begin{split}
        1 &= \left(\frac{d-1}{n+d-1}+1-\frac{d-1}{n+d-1}\right)^{n+d-1}\\
        &=\sum_{i=0}^{n+d-1}\binom{n+d-1}{i}\left(\frac{d-1}{n+d-1}\right)^{i}\left(1-\frac{d-1}{n+d-1}\right)^{(n+d-1)-i}\\
        &\geq \binom{n+d-1}{d-1} \left(\frac{d-1}{n+d-1}\right)^{d-1}\left(1-\frac{d-1}{n+d-1}\right)^{(n+d-1)-(d-1)}\\
        &=\binom{n+d-1}{d-1} 2^{-(n+d-1)H_2\left(\frac{d-1}{n+d-1}\right)}.
    \end{split}
    \end{equation}
    Thus,
    \begin{equation}
        c(d-1)\leq \log\binom{n+d-1}{n} = \log\binom{n+d-1}{d-1} \leq (n+d-1)H_2\left(\frac{d-1}{n+d-1}\right).
    \end{equation}
    Let $x=\frac{n}{d-1}>0$.
    We have
    \begin{equation}
        g(x) = (1+x)H_2\left(\frac{1}{1+x}\right)\geq c.
    \end{equation}
    Note that the bosonic entropy function $g(x)$ is monotonically increasing, since it has derivative $\log(1+1/x)>0$ for all $x>0$. Therefore, we have $x\geq g^{-1}(c)$ and
    \begin{equation}
        n\geq g^{-1}(c)(d-1).
    \end{equation}
    This concludes the proof.
\end{proof}

Now we are ready to prove the lower bound in Theorem~\ref{thm:channel-learning}.

\begin{proof}[Proof of the lower bound in Theorem~\ref{thm:channel-learning}]
    Consider any quantum algorithm that learns $\Phi$ to $\epsilon=\Theta(1)$ error with success probability at least $2/3$ using $n$ queries.
    It is allowed to query $\Phi$ in an arbitrary way (e.g.\ query the inverse and controlled versions of $\Phi$ if they exist, or with indefinite causal order), as in Theorem~\ref{lem:limitation}.
    We take a maximal $3\epsilon$-packing net $\mathcal{M}=\{\Phi_x\}_{x=1}^{|\mathcal{M}|}$ in diamond norm over the set of channels with input dimension $d_A$, output dimension $d_B$, and rank $r$.
    Lemma~\ref{lem:packing-channel} asserts that the cardinality of this net satisfies
    \begin{equation}
        \log|\mathcal{M}|= \Theta(rd_Ad_B\log(1/\epsilon)).
    \end{equation}
    Now we construct a channel distinguishing algorithm that identifies elements of the net $\mathcal{M}$.
    Specifically, we run the channel learning algorithm that makes $n$ queries to any $\Phi_x\in \mathcal{M}$ and outputs a classical description of a channel $\hat{\Phi}$.
    The learning guarantee implies that with probability at least $2/3$, we have $\|\hat{\Phi}-\Phi_x\|_\diamond\leq \epsilon$.
    Triangle inequality then asserts that for any $x'\in [|\mathcal{M}|], x'\neq x$, we have
    \begin{equation}
	   \|\hat{\Phi}-\Phi_{x'}\|_\diamond\geq \|\Phi_{x}-\Phi_{x'}\|_\diamond-\|\hat{\Phi}-\Phi_{x}\|_\diamond \geq 3\epsilon-\epsilon=2\epsilon > \epsilon = \|\hat{\Phi}-\Phi_{x}\|_\diamond.
    \end{equation}
    This means that the channel in the net that is closest to the estimate $\hat{\Phi}$ is unique and exactly $\Phi_x$ itself.
    We can find this closest channel by brute force enumerating all elements of the net.
    This gives a channel distinguishing algorithm with success probability at least $2/3$.
    The channel distinguishing bound Theorem~\ref{lem:limitation} immediately implies that
    \begin{equation}
	\Theta(rd_Ad_B\log(1/\epsilon))=\frac{1}{2}\log|\mathcal{M}|\leq \log\binom{n+rd_Ad_B-1}{n}.
    \end{equation}
    Using Lemma~\ref{lem:binomial}, we arrive at
    \begin{equation}
	n\geq \Omega(rd_Ad_B),
    \end{equation}
    as desired.
    This completes the proof of Theorem~\ref{thm:channel-learning}.
\end{proof}

\section{Conclusion}\label{sec_con}
In this work, we introduce the random Stinespring superchannel, a channel-level analogue of random purification for quantum states. This procedure enables the conversion of multiple parallel uses of an arbitrary quantum channel into equally many parallel uses of the same uniformly random Stinespring isometry, using universal and efficiently implementable encoding and decoding operations. Our proof combines techniques from quantum Shannon theory~\cite{Chiribella2008} to establish the existence of such encoding and decoding operations with representation-theoretic tools based on Schur–Weyl duality to construct an explicit and efficient circuit that realises it.

Beyond its conceptual relevance, the random Stinespring superchannel has concrete
implications for quantum Shannon theory and quantum learning theory. On the quantum Shannon theory side, we show that it yields channel-level extensions of Uhlmann's theorem for quantum divergences~\cite{Mazzola_2025, Fang2025-variational,random_pur_simple}. On the quantum learning theory side, it implies that tomography of quantum channels reduces to tomography of isometries, leading to the recently established upper bounds on the query complexity of quantum channel learning~\cite{AMele2025, chen2025quantumchanneltomographyestimation}. 
As a complementary
result, we derive an improved lower bound on the query complexity that holds even
for the most general classes of queries, including those with inverse and controlled queries and indefinite causal order. 
This is shown by developing a lower bound technique that is purely algebraic and reinforces the simple intuition from dimension counting, completely circumventing the heavy representation theory machinery previously used for unitaries.
We expect that this proof strategy can be applied to quantum channels with other parameterisations beyond bounded rank.
Taken together, these results establish that the optimal query complexity for tomography of quantum channels with input dimension $d_A$, output dimension $d_B$, and Choi rank $r$ scales as $\Theta(d_A d_B r)$, without additional logarithmic factors. In particular, this shows that the upper bound obtained in~\cite{AMele2025} and later reproved in~\cite{chen2025quantumchanneltomographyestimation} is indeed optimal.


We expect our efficient construction of the random Stinespring superchannel to have applications in other fields beyond quantum learning theory and quantum Shannon theory. For example, it may have applications in quantum thermodynamics, specifically in designing quantum thermodynamic protocols by reducing many copies of mixed states to pure states, where energy-optimal and provably-efficient thermodynamic protocols have been developed~\cite{zhao2025learning}.

An intriguing open problem concerns the adaptive setting. Specifically, it remains unclear whether $n$ uses of a quantum channel can be converted into $n$ possibly adaptive uses of a randomly chosen Stinespring isometry associated with the
channel. Another promising open direction is whether, in the same spirit as
in~\cite{WalterWitteveen_2025, cv_purification}, where a protocol is introduced to convert $n$ copies of a Gaussian mixed state into $n$ copies of a randomly chosen Gaussian purification, one can convert $n$ queries of a Gaussian bosonic or fermionic channel into $n$ queries of a randomly chosen Gaussian Stinespring isometry. Such a result would have direct applications to bounding the query complexity of learning Gaussian channels, which has currently been done only in the special case of Gaussian unitary channels~\cite{Gauss_unitary_learning}.

\subsection*{Acknowledgments}
We are grateful to Lennart Bittel, Hsin-Yuan Huang, Iman Marvian, Antonio Anna Mele, and John Wright for inspiring discussions. In particular, we are deeply grateful to Lennart Bittel: early in this project, we had arrived at an incorrect argument purporting to rule out the existence of a random Stinespring superchannel; then his careful feedback revealed the flaw in that reasoning and prompted us to revisit the problem, ultimately leading to the results presented here. MF thanks Giacomo De Palma for his kind hospitality at the University of Bologna, where part of this work was done. FG, FAM, and LL acknowledge financial support from the European Union (ERC StG ETQO, Grant Agreement no.\ 101165230).
The Institute for Quantum Information and Matter is an NSF Physics Frontiers Center (PHY-2317110).

\subsection*{Data Availability Statement}
This is a purely mathematical work and no data was created or analysed in this study.

\bibliography{biblio}

@article{Fang2025-variational,
  title={Variational expressions and {U}hlmann theorem for measured divergences}, 
  author={Fang, Kun and Fawzi, Hamza and Fawzi, Omar},
  year={2025},
  journal={Preprint arXiv:2502.07745},
  doi={10.48550/arXiv.2502.07745}
}

@article{tomamichel12smooth_tutorial,
   author = {Tomamichel, Marco},
   city = {Singapore},
   journal = {QCrypt 2012},
   month = {September},
   title = {Smooth entropies: A Tutorial With Focus on Applications in Cryptography},
   url = {https://marcotom.info/files/qcrypt2012.pdf},
   year = {2012}
}

@article{Schuster2025,
	author = {Thomas Schuster and Jonas Haferkamp and Hsin-Yuan Huang},
	doi = {10.1126/science.adv8590},
	eprint = {https://www.science.org/doi/pdf/10.1126/science.adv8590},
	journal = {Science},
	number = {6755},
	pages = {92-96},
	title = {Random unitaries in extremely low depth},
	url = {https://www.science.org/doi/abs/10.1126/science.adv8590},
	volume = {389},
	year = {2025}}

@BOOK{Hayashi2016_grouptheoretic,
  title     = "A group theoretic approach to quantum information",
  author    = "Hayashi, Masahito",
  publisher = "Springer International Publishing",
  year      =  2016,
  address   = "Basel, Switzerland",
}

@inproceedings{Fei2024QuantumAlgorithm,
  author    = {Fei, Jiani and Timmerman, Sydney and Hayden, Patrick},
  title     = {Quantum Algorithm for Reducing Induced Representations with Applications to Port-based Teleportation},
  booktitle = {Quantum Information Processing (QIP 2024)},
  address   = {Taipei, Taiwan},
  month     = jan,
  year      = {2024},
  note      = {Contributed talk, Jan 15, 16:30--17:00},
  url       = {https://www.youtube.com/watch?v=PhoEYpTXHqI}
}

@article{zhao2025learning,
  title={Learning to erase quantum states: thermodynamic implications of quantum learning theory},
  author={Zhao, Haimeng and Zhang, Yuzhen and Preskill, John},
  journal={Preprint arXiv:2504.07341},
  year={2025}
}

@article{mele_introduction_2024,
	title = {Introduction to {H}aar Measure Tools in Quantum Information: A Beginner’s Tutorial},
    author = {Mele, Antonio Anna},
    journal = {Quantum},
	url = {https://quantum-journal.org/papers/q-2024-05-08-1340/},
    volume = {8},
    pages = {1340},
    year = {2024}
}

@article{bacon_efficient_2006,
	title = {Efficient {Quantum} {Circuits} for {Schur} and {Clebsch}-{Gordan} {Transforms}},
	volume = {97},
	url = {https://link.aps.org/doi/10.1103/PhysRevLett.97.170502},
	doi = {10.1103/PhysRevLett.97.170502},
	number = {17},
	urldate = {2025-12-21},
	journal = {Physical Review Letters},
	author = {Bacon, Dave and Chuang, Isaac L. and Harrow, Aram W.},
	year = {2006},
	pages = {170502}
}

@article{burchardt_high-dimensional_2025,
	title = {High-dimensional quantum {Schur} transforms},
	url = {http://arxiv.org/abs/2509.22640},
	doi = {10.48550/arXiv.2509.22640},
	urldate = {2025-12-21},
	publisher = {arXiv},
	author = {Burchardt, Adam and Fei, Jiani and Grinko, Dmitry and Larocca, Martin and Ozols, Maris and Timmerman, Sydney and Visnevskyi, Vladyslav},
	year = {2025},
	journal = {Preprint arXiv:2509.22640}
}

@book{hayashi_group_2017,
	title = {Group {Representation} for {Quantum} {Theory}},
	copyright = {http://www.springer.com/tdm},
	isbn = {978-3-319-44904-3 978-3-319-44906-7},
	url = {http://link.springer.com/10.1007/978-3-319-44906-7},
	urldate = {2025-12-21},
	publisher = {Springer International Publishing},
	author = {Hayashi, Masahito},
	year = {2017},
	doi = {10.1007/978-3-319-44906-7},
}

@inproceedings{beals_quantum_1997,
	address = {New York, NY, USA},
	series = {{STOC} '97},
	title = {Quantum computation of {Fourier} transforms over symmetric groups},
	isbn = {978-0-89791-888-6},
	url = {https://doi.org/10.1145/258533.258548},
	doi = {10.1145/258533.258548},
	urldate = {2025-12-21},
	booktitle = {Proceedings of the twenty-ninth annual {ACM} symposium on {Theory} of computing},
	author = {Beals, Robert},
	year = {1997},
	pages = {48--53},
}

@article{collins_weingarten_2022,
	title = {The {Weingarten} Calculus},
	volume = {69},
	issn = {0002-9920, 1088-9477},
	url = {http://arxiv.org/abs/2109.14890},
	doi = {10.1090/noti2474},
	number = {05},
	urldate = {2025-12-21},
	journal = {Notices of the American Mathematical Society},
	author = {Collins, Benoit and Matsumoto, Sho and Novak, Jonathan},
	year = {2022},
	pages = {1}
}

@article{collins_integration_2006,
	title = {Integration with Respect to the {Haar} Measure on Unitary, Orthogonal and Symplectic Group},
	volume = {264},
	issn = {1432-0916},
	url = {https://doi.org/10.1007/s00220-006-1554-3},
	doi = {10.1007/s00220-006-1554-3},
	number = {3},
	journal = {Communications in Mathematical Physics},
	author = {Collins, Benoît and Śniady, Piotr},
	year = {2006},
	pages = {773--795}
}

@article{harrow_approximate_2023,
	title = {Approximate orthogonality of permutation operators, with application to quantum information},
	volume = {114},
	issn = {1573-0530},
	url = {http://arxiv.org/abs/2309.00715},
	doi = {10.1007/s11005-023-01744-1},
	number = {1},
	urldate = {2025-12-21},
	journal = {Letters in Mathematical Physics},
	author = {Harrow, Aram W.},
	year = {2023},
	pages = {1}
}

@article{kostenberger_weingarten_2021,
	title = {Weingarten Calculus},
	url = {http://arxiv.org/abs/2101.00921},
	doi = {10.48550/arXiv.2101.00921},
	author = {Köstenberger, Georg},
	year = {2021},
	journal = {Preprint arXiv:2101.00921}
}

@article{kawano_quantum_2016,
	series = {Special issue on the conference {ISSAC} 2014: {Symbolic} computation and computer algebra},
	title = {Quantum {Fourier} transform over symmetric groups — improved result},
	volume = {75},
	issn = {0747-7171},
	url = {https://www.sciencedirect.com/science/article/pii/S0747717115001157},
	doi = {10.1016/j.jsc.2015.11.016},
	urldate = {2025-12-21},
	journal = {Journal of Symbolic Computation},
	author = {Kawano, Yasuhito and Sekigawa, Hiroshi},
	year = {2016},
	pages = {219--243},
}

@article{moore_symmetric_2008,
	title = {The Symmetric Group Defies Strong {Fourier} Sampling},
	volume = {37},
	issn = {0097-5397},
	url = {https://epubs.siam.org/doi/abs/10.1137/050644896},
	doi = {10.1137/050644896},
	number = {6},
	urldate = {2025-12-21},
	journal = {SIAM Journal on Computing},
	author = {Moore, Cristopher and Russell, Alexander and Schulman, Leonard J.},
	year = {2008},
	pages = {1842--1864}
}

@phdthesis{harrow_applications_2005,
	author = {Harrow, Aram W.},
	title = {Applications of coherent classical communication and the {Schur} transform to quantum information theory},
    school = {Massachusetts Institute of Technology},
	year = {2005},
    note = {Preprint arXiv:quant-ph/0512255},
	doi = {10.48550/arXiv.quant-ph/0512255}
}

@book{goodman_symmetry_2009,
	address = {New York, NY},
	series = {Graduate {Texts} in {Mathematics}},
	title = {Symmetry, {Representations}, and {Invariants}},
	volume = {255},
	copyright = {https://www.springernature.com/gp/researchers/text-and-data-mining},
	isbn = {978-0-387-79851-6 978-0-387-79852-3},
	url = {https://link.springer.com/10.1007/978-0-387-79852-3},
	language = {en},
	urldate = {2025-12-21},
	publisher = {Springer},
	author = {Goodman, Roe and Wallach, Nolan R.},
	year = {2009},
	doi = {10.1007/978-0-387-79852-3},
}

@article{krovi_efficient_2019,
	title = {An efficient high dimensional quantum {Schur} transform},
	volume = {3},
	url = {https://quantum-journal.org/papers/q-2019-02-14-122/},
	doi = {10.22331/q-2019-02-14-122},
	language = {en-GB},
	urldate = {2025-12-21},
	journal = {Quantum},
	author = {Krovi, Hari},
	month = feb,
	year = {2019},
	note = {Publisher: Verein zur Förderung des Open Access Publizierens in den Quantenwissenschaften},
	pages = {122},
}

@article{bavaresco2022unitary,
  title={Unitary channel discrimination beyond group structures: Advantages of sequential and indefinite-causal-order strategies},
  author={Bavaresco, Jessica and Murao, Mio and Quintino, Marco T{\'u}lio},
  journal={Journal of Mathematical Physics},
  volume={63},
  number={4},
  year={2022},
  publisher={AIP Publishing}
}

@article{rosenthal2024quantum,
  title={Quantum channel testing in average-case distance},
  author={Rosenthal, Gregory and Aaronson, Hugo and Subramanian, Sathyawageeswar and Datta, Animesh and Gur, Tom},
  journal={Preprint arXiv:2409.12566},
  year={2024}
}

@inproceedings{haah2023query,
  title={Query-optimal estimation of unitary channels in diamond distance},
  author={Haah, Jeongwan and Kothari, Robin and O’Donnell, Ryan and Tang, Ewin},
  booktitle={2023 IEEE 64th Annual Symposium on Foundations of Computer Science (FOCS)},
  pages={363--390},
  year={2023}
}

@article{zhao2024learning,
  title={Learning quantum states and unitaries of bounded gate complexity},
  author={Zhao, Haimeng and Lewis, Laura and Kannan, Ishaan and Quek, Yihui and Huang, Hsin-Yuan and Caro, Matthias C},
  journal={PRX Quantum},
  volume={5},
  number={4},
  pages={040306},
  year={2024},
  publisher={APS}
}

@article{barthel2018fundamental,
  title={Fundamental limitations for measurements in quantum many-body systems},
  author={Barthel, Thomas and Lu, Jianfeng},
  journal={Physical Review Letters},
  volume={121},
  number={8},
  pages={080406},
  year={2018},
  publisher={APS}
}

@article{kretschmann2008information,
  title={The information-disturbance tradeoff and the continuity of {S}tinespring's representation},
  author={Kretschmann, Dennis and Schlingemann, Dirk and Werner, Reinhard F.},
  journal={IEEE transactions on information theory},
  volume={54},
  number={4},
  pages={1708--1717},
  year={2008},
  publisher={IEEE}
}

@article{szarek1997metric,
  title={Metric entropy of homogeneous spaces},
  author={Szarek, Stanis{\l}aw J.},
  journal={Preprint arXiv:math/9701213},
  year={1997}
}

@article{huang2021information,
  title={Information-theoretic bounds on quantum advantage in machine learning},
  author={Huang, Hsin-Yuan and Kueng, Richard and Preskill, John},
  journal={Physical Review Letters},
  volume={126},
  number={19},
  pages={190505},
  year={2021},
  publisher={APS}
}

@article{beals2001quantum,
  title={Quantum lower bounds by polynomials},
  author={Beals, Robert and Buhrman, Harry and Cleve, Richard and Mosca, Michele and De Wolf, Ronald},
  journal={Journal of the ACM (JACM)},
  volume={48},
  number={4},
  pages={778--797},
  year={2001},
  publisher={ACM New York, NY, USA}
}

@article{chen2025quantumchanneltomographyestimation,
  title   = {Quantum channel tomography and estimation by local test},
  author  = {Chen, Kean and Yu, Nengkun and Zhang, Zhicheng},
  journal = {arXiv preprint arXiv:2512.13614},
  year    = {2025},
  url     = {https://arxiv.org/abs/2512.13614}
}

@article{random_pur_simple,
      title={Random purification channel made simple}, 
      author={Girardi, Filippo and Mele, Francesco Anna and Lami, Ludovico},
      year={2025},
      journal={Preprint arXiv:2511.23451},
      url={https://arxiv.org/abs/2511.23451}
}

@article{AMele2025,
      title={Optimal learning of quantum channels in diamond distance}, 
      author={Mele, Antonio Anna and Bittel, Lennart},
      year={2025},
      journal={Preprint arXiv:2512.10214},
      url={https://arxiv.org/abs/2512.10214} 
}

@article{cv_purification,
  title         = {Random purification channel for passive {G}aussian bosons},
  author        = {Mele, Francesco Anna and Girardi, Filippo and Chen, Senrui and Fanizza, Marco and Lami, Ludovico},
  journal       = {Preprint arXiv:2512.16878},
  year          = {2025},
  archivePrefix = {arXiv},
  eprint        = {2512.16878},
  primaryClass  = {quant-ph},
  url           = {https://arxiv.org/abs/2512.16878}
}

@article{Gauss_unitary_learning,
  title   = {Efficient learning of bosonic {G}aussian unitaries},
  author  = {Fanizza, Marco and Iyer, Vishnu and Lee, Junseo and Mele, Antonio A. and Mele, Francesco A.},
  year    = {2025},
  journal = {Preprint arXiv:2510.05531},
  url     = {https://arxiv.org/abs/2510.05531}
}

@article{Utsumi2025,
  title         = {Quantum algorithms for {U}hlmann transformation},
  author        = {Utsumi, Takeru and Nakata, Yoshifumi and Wang, Qisheng and Takagi, Ryuji},
  journal       = {Preprint arXiv:2509.03619},
  year          = {2025},
  archivePrefix = {arXiv},
  eprint        = {2509.03619},
  primaryClass  = {quant-ph},
  url           = {https://arxiv.org/abs/2509.03619}
}

@article{WalterWitteveen_2025,
  title   = {A Random Purification Channel for Arbitrary Symmetries with Applications to Fermions and Bosons},
  author  = {Walter, Michael and Witteveen, Freek},
  journal = {Preprint arXiv:2512.15690},
  year    = {2025},
  archivePrefix = {arXiv},
  eprint  = {2512.15690},
  primaryClass = {quant-ph}
}

@article{Mazzola_2025,
   title={Uhlmann’s Theorem for Relative Entropies},
   volume={71},
   ISSN={1557-9654},
   url={http://dx.doi.org/10.1109/TIT.2025.3591775},
   DOI={10.1109/tit.2025.3591775},
   number={9},
   journal={IEEE Transactions on Information Theory},
   publisher={Institute of Electrical and Electronics Engineers (IEEE)},
   author={Mazzola, Giulia and Sutter, David and Renner, Renato},
   year={2025},
   pages={7039–7051} }

@article{pelecanos2025,
  title={Mixed state tomography reduces to pure state tomography},
  author={Pelecanos, Angelos and Spilecki, Jack and Tang, Ewin and Wright, John},
  journal={Preprint arXiv:2511.15806},
  year={2025},
  url={https://arxiv.org/abs/2511.15806}
}

@article{tang2025,
  title={Conjugate queries can help},
  author={Tang, Ewin and Wright, John and Zhandry, Mark},
  journal={Preprint arXiv:2510.07622},
  year={2025},
  url={https://arxiv.org/abs/2510.07622}
}

@book{NC,
  title={Quantum Computation and Quantum Information: 10th Anniversary Edition},
  author={Nielsen, Micheal A. and Chuang, Isaac L.},
  year={2010},
  publisher={Cambridge University Press},
  address={Cambridge}
}

@article{Stinespring,
  title={Positive functions on {C*}-algebras},
  author={Stinespring, William Forrest},
  journal={Proceedings of the American Mathematical Society},
  volume={6},
  number={2},
  pages={211--216},
  year={1955},
  publisher={JSTOR}
}

@article{squashed,
author = {Christandl, Matthias and Winter, Andreas},
title = {Squashed entanglement: An additive entanglement measure},
journal = {Journal of Mathematical Physics},
volume = {45},
number = {3},
pages = {829--840},
year = {2004},
doi = {https://doi.org/10.1063/1.1643788}
}

@article{Umegaki1962,
  author = "Umegaki, Hisaharu",
  doi = "10.2996/kmj/1138844604",
  journal = "Kodai Mathematical Seminar Reports",
  number = "2",
  pages = "59--85",
  publisher = "Tokyo Institute of Technology, Department of Mathematics",
  title = "{Conditional expectation in an operator algebra. IV. Entropy and information}",
  volume = "14",
  year = "1962"
}

@article{Donald1986,
author="Donald, Matthew J.",
title="On the relative entropy",
journal="Communications in Mathematical Physics",
year="1986",
volume="105",
number="1",
pages="13--34",
issn="1432-0916",
doi="10.1007/BF01212339"
}

@article{newRenyi,
author = {M\"uller-Lennert, Martin and Dupuis, Fr\'{e}d\'{e}ric and Szehr, Oleg and Fehr, Serge and Tomamichel, Marco},
title = {On quantum {R}{\'e}nyi entropies: A new generalization and some properties},
journal = {Journal of Mathematical Physics},
volume = {54},
number = {12},
pages = {122203},
year = {2013}
}

@article{Wilde2014,
  author = {Wilde, Mark M. and Winter, Andreas and Yang, Dong},
  journal = {Communications in Mathematical Physics},
  number = {2},
  pages = {593--622},
  title = {Strong Converse for the Classical Capacity of Entanglement-Breaking and {H}adamard Channels via a Sandwiched {R{\'e}nyi} Relative Entropy},
  volume = {331},
  year = {2014}
}

@article{Chiribella2008,
	doi = {10.1209/0295-5075/83/30004},
	year = {2008},
	publisher = {{IOP} Publishing},
	volume = {83},
	number = {3},
	pages = {30004},
	author = {Chiribella, Giulio and D'Ariano, Giacomo M. and Perinotti, Paolo},
	title = {Transforming quantum operations: Quantum supermaps},
	journal = {{EPL}}
}

@article{Fang2020,
  title = {Chain Rule for the Quantum Relative Entropy},
  author = {Fang, Kun and Fawzi, Omar and Renner, Renato and Sutter, David},
  journal = {Physical Review Letters},
  volume = {124},
  issue = {10},
  pages = {100501},
  numpages = {6},
  year = {2020},
  publisher = {American Physical Society},
  doi = {10.1103/PhysRevLett.124.100501}
}

@article{BBM,
  title = {Quantum cryptography without {B}ell's theorem},
  author = {Bennett, C. H. and Brassard, G. and Mermin, N. D.},
  journal = {Phys. Rev. Lett.},
  volume = {68},
  issue = {5},
  pages = {557--559},
  numpages = {0},
  year = {1992},
  publisher = {American Physical Society},
  doi = {10.1103/PhysRevLett.68.557},
  url = {https://link.aps.org/doi/10.1103/PhysRevLett.68.557}
}

\end{document}